\documentclass[a4paper,fleqn,numbers]{cas-sc}

\usepackage{amsmath}
\usepackage{amsfonts}
\usepackage{amssymb}
\usepackage{subcaption}
\usepackage{graphicx}
\usepackage{arydshln}

\usepackage[numbers]{natbib}

\def\tsc#1{\csdef{#1}{\textsc{\lowercase{#1}}\xspace}}
\tsc{WGM}
\tsc{QE}

\newtheorem{theorem}{Theorem}
\newtheorem{corollary}{Corollary}[theorem]
\newtheorem{assumption}{Assumption}

\newtheorem{propo}{Proposition}
\newdefinition{rmk}{Remark}
\newproof{pf}{Proof}

\begin{document}
\let\WriteBookmarks\relax
\def\floatpagepagefraction{1}
\def\textpagefraction{.001}

\shorttitle{}    

\shortauthors{}  

\title [mode = title]{Synchronization of directed hypergraphs with heterogeneities via dynamic coupling}  



%

\author[1]{Aiwin {Thomas Vadakkan}}[orcid=0009-0000-5493-8385]





\credit{Conceptualization, Methodology, Writing, Software}

\affiliation[1]{organization={University of Naples Federico II},
            addressline={Via Claudio, 21}, 
            city={Naples},
            postcode={80125}, 
            country={Italy}}

\author[1]{Pietro {De Lellis}}
\cormark[1]

\ead{pietro.delellis@unina.it}


\credit{Conceptualization, Supervision, Methodology, Writing}

\author[1,2]{Mario {di Bernardo}}
\cormark[1]

\ead{mario.dibernardo@unina.it}


\credit{Conceptualization, Supervision, Methodology, Writing}

\affiliation[2]{organization={Scuola Superiore Meridionale},
            addressline={Via Mezzocannone, 4}, 
            city={Naples},
            postcode={80138}, 
            country={Italy}}
            
\cortext[1]{Corresponding authors}



\begin{abstract}
Many real-world networks involve interactions among three or more agents that cannot be reduced to pairwise coupling, making hypergraphs a natural modeling framework. In this work, we study complete synchronization in directed hypergraphs of nonlinear agents with parameter mismatches under dynamic diffusive coupling. Although proportional-integral coupling schemes are known to achieve consensus in heterogeneous linear networks, their ability to enforce complete synchronization in nonlinear systems is more limited, generally yielding only bounded synchronization. We show that complete synchronization can be attained only when the effect of parameter mismatches is structurally equivalent, in the transverse dynamics, to a constant disturbance. Under this condition, we derive invariance and local stability conditions for the synchronization manifold and develop a Master Stability Function framework for directed hypergraphs with dynamic coupling. The proposed approach accounts for distinct proportional and integral hypergraph layers and provides spectral criteria for predicting synchronization regions. Numerical simulations on directed hypergraphs of Lorenz oscillators validate the theoretical predictions and show how the integral action can compensate destabilizing effects induced by the proportional layer. We further demonstrate the applicability of the framework to a pinning-control problem in nonlinear opinion dynamics, where dynamic diffusive coupling achieves complete leader-follower consensus in the presence of heterogeneity.

\end{abstract}


\maketitle

\section{Introduction}\label{Introduction}

Complex systems are frequently modeled as networks, where nodes represent individual agents and edges denote their mutual interactions~\cite{strogatz_exploring_2001}. From the local interactions between interconnected agents, collective behaviors emerge, including e.g.  synchronization~\cite{pikovsky_synchronization_2001}, amplitude chimeras~\cite{scholl_partial_2022}, and consensus~\cite{saber_consensus_2003}. Synchronization and consensus, characterized by the asymptotic convergence of agent trajectories to a time-varying solution or equilibrium, respectively, remain the most extensively investigated behaviors due to their critical role in diverse engineering and biological applications. 
Synchronization is fundamental to neural network dynamics~\cite{makovkin_synchronization_2021,singer_neuronal_2018}, as exemplified by hub-driven remote synchronization, where spatially distinct cortical regions achieve coherence via intermediary hubs~\cite{vlasov_hub-driven_2017}. Beyond neuroscience, these phenomena govern diverse biological systems, ranging from the emergence of flocking in bird populations~\cite{sar_flocking_2023} to the collective firing of fireflies~\cite{mccrea_model_2022}. Consensus protocols are equally vital in engineering contexts such as frequency synchronization in power grids~\cite{hill_power_2006,dorfler_synchronization_2010}, distributed formation control in autonomous multi-robot ensembles~\cite{fax_information_2004} and heating, ventilation and air conditioning (HVAC) systems~\cite{weimer_active_2012}.

While graph-based frameworks have traditionally been employed to model complex systems, it is increasingly evident that many real-world networks exhibit interactions among three or more agents that are fundamentally irreducible to pairwise couplings \cite{boccaletti_structure_2023,Millan2025TopologyHODynamics}, a phenomenon particularly evident in opinion dynamics~\cite{lanchier_stochastic_2013}. Moreover, constraints on sensing and actuation often render traditional diffusive coupling infeasible. For instance, when only an aggregated state of groups of nodes can be measured, the feedback signals take the form of a directed multibody interaction \cite{de_lellis_pinning_2023}. 
These considerations led to the adoption of hypergraphs~\cite{boccaletti_structure_2023} to effectively model such higher-order interactions. The investigation of higher-order network dynamics has gained significant traction in the last decade, with applications also spanning epidemic spreading~\cite{bodo_sis_2016} and social contagion processes~\cite{iacopini_simplicial_2019}.

Beyond neglecting higher-order interactions, the classical literature on synchronization and consensus typically made the simplifying assumption of identical node dynamics. However, in several applications, the assumption of agent homogeneity often fails due to inherent parameter mismatches, environmental noise, and external disturbances. Under these conditions,  diffusive coupling is insufficient to achieve complete synchronization; instead, it yields bounded synchronization, where the synchronization error does not vanish asymptotically but remains uniformly bounded by some constant. For instance, Zhao et al. ~\cite{zhao_global_2012} proved global bounded synchronization of nonidentical nodes given that the average vector field is bounded,  while the case of discontinuous nonidentical nodes was tackled in~\cite{delellis_convergence_2015}. Parallel challenges exist in consensus theory, where diffusive strategies only guarantee a bounded steady-state error in the presence of nodal heterogeneity and disturbances~\cite{zhong_global_2012,wang_bounded_2013,kim_practical_2012}.

To address the limitations of diffusive coupling, dynamic coupling strategies - specifically Proportional-Integral (PI) control schemes - have been proposed. These frameworks augment standard diffusive coupling (that can be viewed as a proportional action) with an integral action to mitigate the effects of disturbances and nodal heterogeneities. For instance, Freeman \textit{et al.}~\cite{freeman_stability_2006} demonstrated that consensus is attainable in homogeneous linear systems subject to constant disturbances via PI control, provided that the proportional and integral topologies are identical and their respective coupling gains are independently tunable. Subsequently, these constraints were relaxed by introducing a distributed PI control strategy that facilitates consensus in networks of heterogeneous linear systems under constant perturbations~\cite{burbano_lombana_multiplex_2016}. Further extensions and applications of PI-based coordination can be found in~\cite{andreasson_distributed_2014,carli_pi_2008,wang_distributed_2019}, spanning from network congestion control~\cite{shang_decentralized_2004} to the stabilization of modern power systems~\cite{andreasson_distributed_2014,abo-elkhair_enhancing_2025}. While dynamic coupling (PI) facilitates consensus in networks subject to disturbances and parameter heterogeneities, its efficacy in achieving complete synchronization within heterogeneous nonlinear networks is restricted; in general, such coupling schemes yield only bounded synchronization~\cite{burbano_lombana_synchronization_2016}.



The present work applies the PI dynamic coupling to cope with nonidentical node dynamics in the presence of higher-order interactions described by directed hypergraphs using the Master Stability Function (MSF) approach, which was originally formulated for undirected networks of homogeneous agents~\cite{pecora_master_1998}. The MSF facilitates the evaluation of the
local transverse stability of the synchronization manifold by determining whether perturbations transverse to the
manifold decay or amplify asymptotically. 
The MSF framework has been significantly extended over the past decades to individually accommodate parameter mismatches~\cite{sun_master_2009}, directed topologies~\cite{nishikawa_synchronization_2006}, multilayer architectures~\cite{del_genio_synchronization_2016}, dynamic coupling on graphs~\cite{burbano_lombana_synchronization_2016}, and higher-order structures~\cite{della_rossa_emergence_2023,Gallo2022DirectedHOI}. However, there is no current MSF framework that simultaneously addresses parameter mismatch, dynamic coupling and directed higher order interactions. Hence, we establish a comprehensive theoretical framework to study the onset of spontaneous or controlled synchronization in hypergraphs governed by dynamic diffusive coupling, with particular emphasis on complete synchronization of agents with parameter mismatches. Specifically, our contributions are: 1) we derive analytical conditions for the invariance and local stability of the synchronization manifold for agents with parameter mismatches interacting via directed hypergraphs, 2) we use the theoretical results to define a MSF framework for directed hypergraphs accommodating both parameter mismatches and dynamic coupling which enables the precise prediction of synchronization regions for arbitrary higher-order network topologies (detailed in Section~\ref{Results}), 3) we test the effectiveness of our approach on testbed numerical applications, namely  Erd\H{o}s-R\'enyi (ER) like directed hypergraphs of Lorenz systems, and a nonlinear opinion dynamics model adapted from \cite{rizzello_pinning_2024}. The numerical analyses confirm the effectiveness of the approach in identifying the conditions on the coupling and control gains for which complete synchronization can be attained, and underline the relevance of the integral action.

The remainder of the paper is organized as follows. Section 2 introduces the notation and preliminary concepts on directed hypergraphs used throughout the paper. Section 3 formulates the synchronization problem for directed hypergraphs with parameter mismatches under dynamic diffusive coupling. Section 4 presents the main theoretical results, including the conditions for invariance of the synchronization manifold and the Master Stability Function framework for assessing its local transverse stability. Section 5 validates the theoretical findings through numerical simulations on networks of Lorenz oscillators and illustrates the applicability of the proposed approach to a pinning-control problem in opinion dynamics. Finally, Section 6 summarizes the main conclusions and outlines possible directions for future research.

\section{Preliminaries} \label{Preliminaries}

Given $n \in \mathbb{Z}^{+}$, $I_n \in \mathbb{R}^{n \times n}$ is the identity matrix, $\mathbf{0}_n \in \mathbb{R}^{n \times n}$ is the matrix of zeros, $0_n \in \mathbb{R}^n$ is the vector of zeros and $1_n \in \mathbb{R}^n$ is the vector of ones. Given a vector $v \in \mathbb{R}^n$, $\mathrm{diag(v)} \in \mathbb{R}^{n \times n}$ represents a diagonal matrix whose $i$th diagonal entry is the $i$th element of the vector $v$. Given $v_1,\dots v_p \in \mathbb{R}^n$, we denote $[v_1,\dots,v_p] \in \mathbb{R}^{n \times p}$ and  $[v_1;\dots;v_p] \in \mathbb{R}^{np}$ to be their horizontal and vertical concatenations, respectively.

Given a complex number $c \in \mathbb{C}$, $\Re(c)$ and $\Im(c)$ correspond to the real and imaginary part of $c$, respectively. Given a matrix $ A \in \mathbb{R}^{n \times n}$, $\mathrm{spec(A)}$ is its spectrum, and $\lambda_{\min}(A)$ and $\lambda_{\max}(A)$ denote the eigenvalues with the  smallest and largest real part, respectively. For any two matrices $M_1 \in \mathbb{R}^{a \times b}$ and $M_2 \in \mathbb{R}^{c \times d}$, $M_1$ $\otimes$ $M_2  \in \mathbb{R}$$^{ac \times bd}$ denotes their Kronecker product. Finally given a vector field $f(a,b): \mathbb{R}^p \times \mathbb{R}^q \rightarrow \mathbb{R}^n$, $\mathbf{J}_a f \in \mathbb{R}^{n \times p}$ and  $\mathbf{J}_b f \in \mathbb{R}^{n \times q}$ are the Jacobian of the vector field $f$ with respect to $a$ and $b$, respectively.

\subsection{Hypergraphs} \label{Preliminaries - Directed hypergraphs}

An undirected hypergraph \cite{boccaletti_structure_2023} is a pair $\mathcal{H} = (\mathcal{V}, \mathcal{E})$ with $\mathcal{V} = \{ {v}_1,\dots,{v}_N\}$ being the set of nodes and $\mathcal{E} = \{h_1,\dots,h_{|\mathcal{E|}}\}$ the family of non-empty subsets of the elements of $\mathcal{V}$. These subsets are termed as hyperedges and they naturally encode higher order interactions i.e  whenever $|h| > 2$. The rank of a  hypergraph is defined as the maximum  cardinality of its hyperedges: $\mathrm{rank}(\mathcal{H}) = \max_{h \in \mathcal{E}}\{|h|\}$. Similarly, the co-rank is defined as the minimum cardinality of its hyperedges: $\mathrm{co}$-$\mathrm{rank}(\mathcal{H}) = \min_{h \in \mathcal{E}}\{|h|\}$. If $\mathrm{rank}(\mathcal{H}) = \mathrm{co}$-$\mathrm{rank}(\mathcal{H}) = k$, then the hypergraph is called k-uniform hypergraph.

A directed hypergraph \cite{gallo_directed_1993} is a pair $(\mathcal{V},\mathcal{E})$, with $\mathcal{V} = \{ {v}_1,\dots,{v}_N\}$ being the set of nodes and $\mathcal{E}$ the set of directed hyperedges. A directed hyperedge $(e)$ is a pair of ordered, disjoint subsets of $\mathcal{V}$ such that the first set denotes the tails, $\mathcal{T}(e)$, and the second set denotes the heads, $\mathcal{H}(e)$, whose cardinality is $|\mathcal{T}(e)|$ and $|\mathcal{H}(e)|$, respectively. Given two subsets ${\mathcal{V}}_1, {\mathcal{V}}_2 \subset \mathcal{V}$,  we define $\mathcal{E}^{{\mathcal{V}}_1,{\mathcal{V}}_2} = \{e \in \mathcal{E}: \mathcal{V}_1 \subseteq \mathcal{T}(e) \wedge \mathcal{V}_2 \subseteq \mathcal{H}(e) \}$ as the set of hyperedges in $\mathcal E$ whose tail set includes $\mathcal V_1$, and whose head set includes $\mathcal V_2$. Finally, we denote $\mathcal{E}^j = \
\{ e \in  \mathcal{E}: v_j \in \mathcal{H}(e) \}$, which is the subset of hyperedges in $\mathcal{E}$ having $v_j$ as its head. 



\section{Problem Statement}\label{Problem Statement}

We consider the problem of achieving complete synchronization in a higher-order network of $N$ nodes with parameter mismatch, whose dynamics is governed by
\begin{equation} \label{Open loop dynamics}
    \frac{dx_i}{dt} = f(x_i,\mu_i)+\sum_{e \in \mathcal{E}_P^i}(\sigma_e)_P g\left( x_e^{\mathcal{T}} \alpha_e - x_e^{\mathcal{H}}\beta_e \right) + u_i,\qquad i=1,\ldots,N,
\end{equation}
where $f:\mathbb{R}^n \times \mathbb{R}^q \rightarrow\mathbb{R}^n$ is the vector field describing the individual dynamics, $x_i \in \mathbb{R}^n$ is the state vector of node $i$, and $\mu_i \in \mathbb{R}^q$ is the parameter vector associated to node $i$, and $u_i$ is the distributed control input acting on node $i$. The nodes are coupled through a proportional, hyperdiffusive coupling protocol over a directed hypergraph $\{\mathcal V,\mathcal E_P\}$, with $\mathcal V$ being the set of $N$ coupled nodes, and $\mathcal E_P$ is the set of hyperedges interconnecting them. The interactions between interconnected nodes take place through a hyperdiffusive coupling protocol \cite{de_lellis_pinning_2023}, where $g: \mathbb{R}^n \rightarrow\mathbb{R}^n$ is a generic nonlinear coupling function, $\mathcal{E}_P^i$ denotes the set of all the hyperedges in $\mathcal E_p$ with the $i$th node among its heads, and matrices $x_e^{\mathcal{T}}\in \mathbb{R}^{n \times |\mathcal{T}(e)|}$ and $x_e^{\mathcal{H}} \in \mathbb{R}^{n \times |\mathcal{H}(e)|}$ collect the states of the tail and head nodes of hyperedge $e$ and horizontally stack them, respectively; $(\sigma_e)_P$ is the coupling strength of the proportional action associated to hyperedge $e$, whereas the vectors $\alpha_e \in \mathbb{R}^{|\mathcal{T}(e)|}$ and $\beta_e \in \mathbb{R}^{|\mathcal{H}(e)|}$ vertically stack the weights associated with the tail and head nodes of hyperedge $e$, respectively, and are such that  $\alpha_e^{\mathrm{T}} 1_{|\mathcal{T}(e)|} = \beta_e^{\mathrm{T}} 1_{|\mathcal{H}(e)|} = 1$.

The hyperdiffusive coupling protocol in $\eqref{Open loop dynamics}$ implies that each head of a hyperedge receives a signal which is a function of the difference between a convex combination of the states of the tails and a convex combination of the states of the heads of that hyperedge. Moreover, we assume that $g(0) =0$, making this protocol synchronization noninvasive \cite{gambuzza_stability_2021}, a precondition for the invariance of the synchronization manifold.
However, its invariance is prevented by the parameter mismatches between the nodes, thus requiring an appropriate design of the control input $u_i$ to possibly enforce complete synchronization. 

In this work, in order to achieve complete synchronization in such hypergraphs with nonidentical nodes, we propose the use of dynamic diffusive coupling, thereby selecting $u_i$ as
\begin{equation}
    u_i = \sum_{e \in \mathcal{E}_I^i} (\sigma_e)_I\int_0^t g\left( x_e^{\mathcal{T}} \alpha_e - x_e^{\mathcal{H}}\beta_e \right)d\tau,
\end{equation}
where $(\sigma_e)_I$ is the non-zero coupling strength of the integral action associated to hyperedge $e$; $\mathcal E_I$ is the hyperedge set of the directed hypergraph that describes the distributed integral interactions between the nodes in $\mathcal V$, and $\mathcal{E}_I^i$ is the set of all hyperedges in $\mathcal E_I$ with node $i$ among its heads. Hence, with this choice of the control input, the dynamics of the closed loop system can be written as
\begin{align}\label{full dynamics}
\frac{dx_i}{dt} &= f(x_i,\mu_i)+\sum_{e \in \mathcal{E}_P^i}(\sigma_e)_P g\left( x_e^{\mathcal{T}} \alpha_e - x_e^{\mathcal{H}}\beta_e \right) + u_i,  \\
\frac{du_i}{dt} &= \sum_{e \in \mathcal{E}_I^i}(\sigma_e)_Ig\left( x_e^{\mathcal{T}} \alpha_e - x_e^{\mathcal{H}}\beta_e \right),\label{eq:udot}
\end{align}
The closed-loop network can be then represented as a two-layer network on the nodes in $\mathcal V$ as illustrated in Fig.~\ref{Fig.1}, which shows how the proportional and integral layer are defined by the hyperedge sets $\mathcal E_P$ and $\mathcal E_I$, respectively.

\begin{figure}
    \centering
    \includegraphics[width=0.6\linewidth]{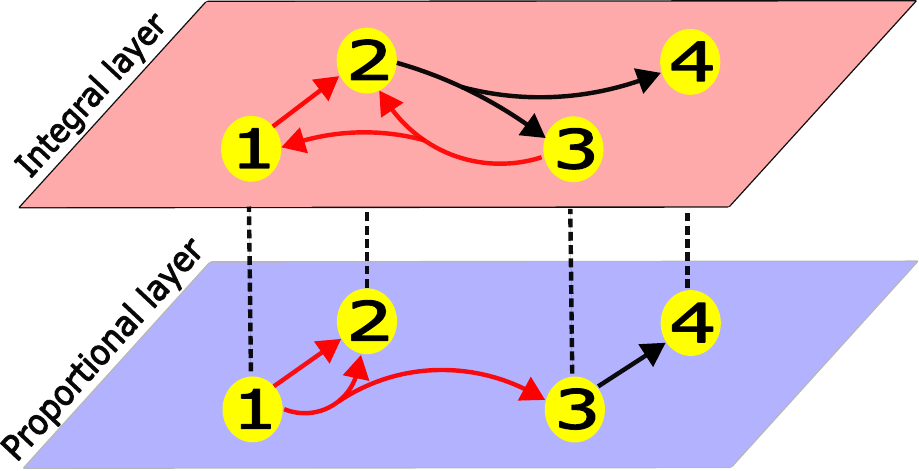}
    \caption{Two-layer directed hypergraph representing the interactions in the closed-loop network \eqref{full dynamics}-\eqref{eq:udot}. The red hyperedges identify all hyperedges having  $2$ as a head, whereby the set of hyperedges influencing node $2$ in the proportional layer is  $\mathcal{E}_P^2 = \left\{ \left( \{1\},\{2\} \right), \left( \{1\},\{2,3\} \right) \right\}$, whereas in the integral layer is $\mathcal{E}_I^2 = \left\{ \left( \{1\},\{2\} \right), \left( \{3\},\{1,2\} \right) \right\}$.}
    \label{Fig.1}
\end{figure}

\section{Results} \label{Results}

Here, we first show a necessary and sufficient condition on the vector field $f$ for the invariance of the synchronization manifold, and then, under the identified condition, we study the local stability of the synchronization manifold.
Let us start by defining the average state trajectory
\begin{equation}
    \quad \bar{x}(t) = \frac{1}{N} \sum_{k=1}^N x_k(t),
\end{equation}
whose dynamics are given by
\begin{equation} \label{average traj dynamics}
    \dot {\bar{x}}(t) = \frac{1}{N} \sum_{k=1}^N \bigg[f(x_k,\mu_k)+\sum_{e \in \mathcal{E}_P^k}(\sigma_e)_P g\left( x_e^{\mathcal{T}} \alpha_e - x_e^{\mathcal{H}}\beta_e \right) + u_k \bigg] \\
\end{equation}
We can now introduce the error of the $i$th node as the deviation of its state from $\bar x(t)$, that is, 
\begin{equation} \label{error}
\delta x_i(t) = x_i(t) - \bar{x}(t)
\end{equation}
By combining \eqref{full dynamics} and \eqref{average traj dynamics}, the error dynamics can be written 
\begin{equation}\label{eq:err_dyn}
    \delta \dot x_i =  f(x_i,\mu_i)+\sum_{e \in \mathcal{E}_P^i}(\sigma_e)_P g\left( x_e^{\mathcal{T}} \alpha_e - x_e^{\mathcal{H}}\beta_e \right) + u_i  - \frac{1}{N} \sum_{k=1}^N \bigg[f(x_k,\mu_k)+\sum_{e \in \mathcal{E}_P^k}(\sigma_e)_P g\left( x_e^{\mathcal{T}} \alpha_e - x_e^{\mathcal{H}}\beta_e \right) + u_k \bigg]
\end{equation}
Next, we recall from \cite[Lemma 1]{rizzello_pinning_2024} that for any hyperedge $e \in \mathcal{E}$, for all $i \in \mathcal{H}(e)$, the following equality holds:
\begin{equation}   \label{lemma 1}
    x_e^{\mathcal{T}} \alpha_e - x_e^{\mathcal{H}} \beta_e = \sum_{j \in \mathcal{T}(e)} (\tilde \alpha_e)_j(x_j-x_i) - \sum_{j \in \mathcal{H}(e)} (\tilde \beta_e)_j(x_j-x_i),
\end{equation}
where the $j$th entry of $ \tilde\alpha_e$ ($\tilde \beta_e$) is 0 if node $j$ is not a tail (head), whereas it is equal to the weight associated to the corresponding tail (head), otherwise. By combining \eqref{eq:err_dyn} with \eqref{lemma 1}, the error dynamics can be rewritten as
\begin{multline}\label{error dynamics 2}
    \delta \dot x_i = f(x_i, \mu_i) - \frac{1}{N} \sum_{k=1}^N  f(x_k, \mu_k)+  \sum_{e \in \mathcal{E}_P^i} \Bigg[(\sigma_e)_p g\left(\sum_{j \in \mathcal{T}(e)} (\tilde \alpha_e)_j(x_j-x_i) - \sum_{j \in \mathcal{H}(e)} (\tilde \beta_e)_j(x_j - x_i)\right) \Bigg] \\ 
    - \frac{1}{N} \sum_{k=1}^N  \Bigg[\sum_{e \in \mathcal{E}_P^k} \Bigg[(\sigma_e)_p g\left(\sum_{j \in \mathcal{T}(e)} (\tilde \alpha_e)_j(x_j  -x_k)  - \sum_{j \in \mathcal{H}(e)} (\tilde \beta_e)_j(x_j - x_k) \right )  \Bigg] \Bigg]  + \delta u_i  
\end{multline}
where $\delta u_i  = u_i - \frac{1}{N} \sum_{k=1}^N u_k $ is the deviation of the control action at the $i$th node from the average control action; the dynamics of $\delta u_i$ can be computed by applying again \cite[Lemma 1]{rizzello_pinning_2024}, thus obtaining
    \begin{align} 
        \delta \dot u_i &= \sum_{e \in \mathcal{E}_I^i} (\sigma_{e})_I \, g \Biggl( \sum_{j \in \mathcal{T}(e)} (\tilde \alpha_e)_j(x_j-x_i) - \sum_{j \in \mathcal{H}(e)} (\tilde \beta_e)_j(x_j - x_i) \Biggr) \nonumber \\
        &\quad - \frac{1}{N} \sum_{k=1}^N \Bigg[\sum_{e \in \mathcal{E}_I^k} (\sigma_{e})_I \, g \Biggl( \sum_{j \in \mathcal{T}(e)} (\tilde \alpha_e)_j(x_j-x_k) - \sum_{j \in \mathcal{H}(e)} (\tilde \beta_e)_j(x_j - x_k) \Biggr) \Bigg]. \label{eq:dy_dt}
    \end{align}

Now assuming that $\delta x$ is small enough, as well as the parameter mismatch with respect to the average parameter vector $\bar \mu=\sum_{i=1}^N \mu_i/N$, we can use first-order Taylor expansions to approximate the dynamics \eqref{error dynamics 2}-\eqref{eq:dy_dt} as
\begin{subequations}
\begin{align}
    \delta \dot x &=  \mathbf{J}_x f(\bar{x}, \bar{\mu}) \delta x_i + \mathbf{J}_\mu f(\bar{x}, \bar{\mu}) \delta \mu_i 
    + \mathbf{J}_x g(0)\Bigg[\sum_{j=1}^N A^P_{ij}(\delta x_j - \delta x_i)  - \frac{1}{N}\sum_{k=1}^N\sum_{j=1}^N A^P_{kj}(\delta x_j - \delta x_k)\Bigg]  + \delta u_i, \label{error dynamics 3} \\
     \delta \dot u_i &= \mathbf{J}_x g(0) \Bigg[\sum_{j=1}^N A^I_{ij}(\delta x_j - \delta x_i)  - \frac{1}{N}\sum_{k=1}^N\sum_{j=1}^N A^I_{kj}(\delta x_j - \delta x_k)\Bigg], \label{error dynamics 4}
\end{align}    
\end{subequations}
where we have used the fact that $\sum_{k=1}^N\delta x_k = 0_n$, $\sum_{k=1}^N\delta\mu_k=0_q$, and where
\begin{align}\label{signed matrix}
        A^P_{ij} &= \sum_{e \in \mathcal{E}_P^i} \sum_{j \in \mathcal{T}(e)} (\sigma_e)_P(\tilde \alpha_e)_j - \sum_{e \in \mathcal{E}_P^i} \sum_{j \in \mathcal{H}(e)} (\sigma_e)_P(\tilde \beta_e)_j \\
    A^I_{ij} &= \sum_{e \in \mathcal{E}_I^i} \sum_{j \in \mathcal{T}(e)} (\sigma_e)_I(\tilde \alpha_e)_j - \sum_{e \in \mathcal{E}_I^i} \sum_{j \in \mathcal{H}(e)} (\sigma_e)_I(\tilde \beta_e)_j 
\end{align} 
are the elements of the adjacency matrices $A^P$ and $A^I$ of the signed graph associated to the directed hypergraph $\{\mathcal V,\mathcal E_P\}$ and $\{\mathcal V,\mathcal E_I\}$, respectively\footnote{For more details on the association between directed hypergraphs and signed graphs, we refer the reader to \cite{de_lellis_pinning_2023,della_rossa_emergence_2023}.}.  
By introducing the Laplacian matrices $L^P$ and $L^I$ associated with $A^P$ and $A^I$, respectively, the error dynamics given in $\eqref{error dynamics 3}$ and $\eqref{error dynamics 4}$ could then be written as follows:
\begin{subequations}
    \begin{align}
    \delta \dot x_i &= \mathbf{J}_x f(\bar{x}, \bar{\mu}) \delta x_i + \mathbf{J}_\mu f(\bar{x}, \bar{\mu}) \delta \mu_i
    - \mathbf{J}_x g(0) \sum_{j=1}^N \widetilde{L}^{P}_{ij} \delta x_j + \delta u_i,\\
    \delta \dot u_i &=  -\mathbf{J}_x g(0)\sum_{j=1}^N \widetilde{L}^{I}_{ij} \delta x_j,
    \end{align}
\end{subequations}
where $\widetilde{L}^{P} = L^P - 1_N  \left[\frac{1}{N}\sum_{k=1}^NL^P_{k1}, \dots ,\frac{1}{N}\sum_{k=1}^NL^P_{kN}\right] $ and  $\widetilde{L}^{I} = L^I - 1_N \left[\frac{1}{N}\sum_{k=1}^NL^I_{k1}, \dots ,\frac{1}{N}\sum_{k=1}^NL^I_{kN}\right]$. We can then define the dynamics of $\delta x = [\delta x_1; \dots ;\delta x_N] \in \mathbb{R}^{nN}$ and $\delta u = [\delta u_1; \dots ;\delta u_N] \in \mathbb{R}^{nN}$ in compact form as 
\begin{subequations} \label{eq:deltaxu_dot}
\begin{align}
    \delta \dot x &= \bigg[I_N \otimes\mathbf{J}_xf(\bar{x},\bar{\mu})-\widetilde{L}^{P}\otimes \mathbf{J}_xg(0)\bigg]\delta x+\bigg[I_N\otimes\mathbf{J}_\mu  f(\bar{x},\bar{\mu})\bigg]\delta \mu + \delta u,\label{eq:deltax_dot} \\
    \delta \dot u &=  -\bigg[\widetilde{L}^{I}\otimes\mathbf{J}_xg(0)\bigg]\delta x,\label{eq:deltau_dot}
\end{align} 
\end{subequations}
where $\otimes$ is the Kronecker product.

Saying that the synchronization manifold $x_1=\ldots=x_N$ is invariant, means that once the network is synchronized, it will keep the synchronous state forever, that is, when $\delta x=0_{nN}$, also $\delta \dot x=0_{nN}$. 

The hyperdiffusive coupling term vanishes at the synchronization manifold because $x_e^{\mathcal {T}}\alpha_e-x_e^{\mathcal H}\beta_e=0$ and $g(0)=0$ whenever $x_1=\ldots=x_N$. However, when the vector fields differ from node to node, the synchronization manifold need not to be invariant, since the nodal vector fields may induce different derivatives even at identical states. The following proposition characterizes the general condition for invariance
\begin{propo}
The synchronization manidold $x_1=\ldots,=x_N=s$ is invariant for the closed-loop system only if, on the manifold,
\begin{equation}
f(s,\mu_i)+u_i=f(s,\mu_j)+u_j,\qquad i,j=1,\ldots,N
\end{equation}
\begin{pf}
Since the coupling protocol is synchronization noninvasive, the proof trivially follows.
\end{pf}
\end{propo}
We now characterize the condition under which the first-order transverse effect of the parameter mismatch can be exactly compensated by the dynamic integral state.

\begin{theorem}[First-order mismatch compensation]
Consider the closed-loop dynamics \((3)\)--\((4)\), and suppose that the state and parameter deviations from
\((\bar{x},\bar{\mu})\) are sufficiently small for the first-order transverse dynamics \eqref{eq:deltaxu_dot} to be valid.
Then the synchronous subspace \(\delta x=0\) is invariant for the first-order transverse dynamics if and only if, along the synchronous trajectory,
\[
\delta u
=
-\left(I_N\otimes J_{\mu}f(\bar{x},\bar{\mu})\right)\delta\mu .
\]
Moreover, since \(\delta \dot{u}=0\) whenever \(\delta x=0\), this cancellation can be maintained along the synchronous trajectory only if
\[
\left(I_N\otimes J_{\mu}f(\bar{x}(t),\bar{\mu})\right)\delta\mu
\]
is constant in time. In particular, this condition is satisfied if \(J_{\mu}f(x,\mu)\) is constant, in which case the parameter mismatch enters the first-order transverse dynamics as an additive constant disturbance.
\end{theorem}

\begin{pf}
Evaluating the first-order transverse dynamics \eqref{eq:deltaxu_dot} on the synchronous subspace \(\delta x=0\) gives
\[
\delta \dot{x}
=
\left(I_N\otimes J_{\mu}f(\bar{x},\bar{\mu})\right)\delta\mu
+
\delta u,
\]
and
\[
\delta \dot{u}=0 .
\]
Therefore, the subspace \(\delta x=0\) is invariant for the first-order transverse dynamics if and only if
\[
\delta \dot{x}=0
\]
whenever \(\delta x=0\). This is equivalent to requiring
\[
\delta u
=
-\left(I_N\otimes J_{\mu}f(\bar{x},\bar{\mu})\right)\delta\mu .
\]

Since \(\delta \dot{u}=0\) on \(\delta x=0\), the value of \(\delta u\) cannot vary while the trajectory remains on the synchronous subspace. Hence, the cancellation condition above can persist along the synchronous trajectory only if
\[
\left(I_N\otimes J_{\mu}f(\bar{x}(t),\bar{\mu})\right)\delta\mu
\]
is constant in time. A sufficient structural condition for this to hold is that \(J_{\mu}f(x,\mu)\) is constant. Under this condition, the first-order contribution of the parameter mismatch is equivalent to an additive constant disturbance, which can be compensated by the integral component of the dynamic diffusive coupling.
\end{pf}
\color{black}





Notice that, the condition $\mathbf{J}_\mu f(x,\mu)$ being constant implies that, to first order, the mismatch enters as an additive constant input. Hence, dynamic diffusive coupling can reject the mismatch when the parameter heterogeneity is indistinguishable from a constant nodal disturbance in the transverse dynamics.
Now that we have characterized the invariance of the synchronization manifold, we can then focus on its local transverse stability. Although a general result cannot be drawn for arbitrary hypergraph topologies of the two layers, local transverse stability can be proved when some assumptions on the two layers are made, and specifically

\begin{assumption}\label{ass:1}
Laplacian matrix $\widetilde L^I$ has a simple zero eigenvalue. Moreover, denoting the elements of the spectra of $\widetilde{L}^P$ and $\widetilde{L}^I$ as $\mathrm{spec}(\widetilde{L}^P) =  \{0,\lambda^P_2,\dots,\lambda^P_N\}$ and $\mathrm{spec}(\widetilde{L}^I) = \{0,\lambda^I_2,\dots, \lambda_N^I\}$, respectively, there exists a transformation $T$  such that $T^{-1}\widetilde{L}^PT = \Lambda^P$ and $T^{-1}\widetilde{L}^IT = \Lambda^I$ with
\begin{equation}\label{eq:lambdas}
\renewcommand{\arraystretch}{1.5}
    \Lambda^P = 
\begin{bmatrix}
    0      & 0               & \dots  & 0                 & 0           \\
    0      & \lambda^P_2     & 0      & \dots             & 0           \\
    0      & \Lambda^P_{32} & \lambda_3^P & \dots        & 0           \\
    \vdots & \vdots          & \ddots & \ddots            & \vdots      \\
    0      & \Lambda^P_{N2} & \dots  & \Lambda^P_{NN-1} & \lambda_N^P
\end{bmatrix},
    \quad 
    \Lambda^I = \begin{bmatrix}
    0      & 0               & \dots  & 0                 & 0           \\
    0      & \lambda^I_2     & 0      & \dots             & 0           \\
    0      & \Lambda^I_{32} & \lambda_3^I & \dots        & 0           \\
    \vdots & \vdots          & \ddots & \ddots            & \vdots      \\
    0      & \Lambda^I_{N2} & \dots  & \Lambda^I_{NN-1} & \lambda_N^I
\end{bmatrix}.
\end{equation}
\end{assumption}
We are now ready to state the following result:
    
\begin{theorem}  \label{Theorem 2}
If Assumption \ref{ass:1} holds,
$\|\mathbf{J}_x g(0)\|$ is bounded,  $\mathbf{J}_\mu f(\bar x, \bar \mu)$ is constant, $\delta u(0)= [-\mathbf{I}_N \otimes\mathbf{J}_\mu f(\bar x, \bar \mu
)]\delta \mu$,
and the maximum Lyapunov exponent (MLE) of
\begin{equation}
\mathcal{M}_i = 
\begin{bmatrix}
\mathbf{J}_x f(\bar{x},\bar{\mu}) - \lambda^P_i \mathbf{J}_x g(0) & -\mathbf{I}_n \\
\lambda^I_i \mathbf{J}_x g(0) & \mathbf{0}_n
\end{bmatrix}.
\end{equation}
is negative for all $\lambda^P_i$ and $\lambda^I_i$ where $i=2, \dots, N$, then the synchronization manifold  under dynamic diffusive coupling is locally transversely stable.

\end{theorem}

\begin{pf}
By applying the change of variables $\eta=(T^{-1} \otimes I_n)\delta x$ and $\gamma = -(T^{-1} \otimes I_n)\delta u$, we rewrite \eqref{eq:deltaxu_dot} as
\begin{subequations} \label{error dynamics 6}
\begin{align}
    \dot \eta &= \bigg[I_N \otimes\mathbf{J}_xf(\bar{x},\bar{\mu})-\Lambda^P\otimes \mathbf{J}_xg(0)\bigg]\eta+\bigg[T^{-1}\otimes\mathbf{J}_\mu  f(\bar{x},\bar{\mu})\bigg]\delta \mu - \gamma,  \\
    \dot \gamma &=  \bigg[\Lambda^I\otimes\mathbf{J}_xg(0)\bigg]\eta,
\end{align} 
\end{subequations}
From \eqref{eq:lambdas}, \eqref{error dynamics 6} can be written as
\begin{subequations}
\begin{align}
    \dot\eta_{1} &= \mathbf{J}_xf(\bar{x},\bar{\mu})\eta_{1}, \label{MSf syn} \\ \nonumber
    \dot \gamma_{1} &= 0_n,\\
    \dot\eta_{2} &= \bigg[\mathbf{J}_xf(\bar{x},\bar{\mu})-\mathbf{J}_xg(0)\lambda^P_2 \bigg]\eta_{2} -\gamma_{2}+ \mathbf{J}_\mu f(\bar{x},\bar{\mu}) w_{2}, \label{MSf 1} \\ \nonumber
    \dot \gamma_{2} &= \mathbf{J}_xg(0)\lambda^I_2\eta_{2}\\
    \dot\eta_{3} &= \bigg[\mathbf{J}_xf(\bar{x},\bar{\mu})- \mathbf{J}_xg(0)\lambda^P_3\bigg]\eta_{3} - \mathbf{J}_xg(0)\Lambda_{32}^P\eta_{2} - \gamma_{3} + \mathbf{J}_\mu f(\bar{x},\bar{\mu})w_{3}, \label{MSf 2} \\  \nonumber
    \dot \gamma_{3} &= \mathbf{J}_xg(0)\lambda^I_3\eta_{3}+\mathbf{J}_xg(0)\Lambda_{32}^I\eta_{2}, \nonumber\\
    &\vdots \nonumber\\
    \dot\eta_{N} &= \mathbf{J}_xf(\bar{x},\bar{\mu})\eta_{N} - \mathbf{J}_xg(0)\bigg[ \Lambda_{N2}^P\eta_{2} + \Lambda_{N3}^P\eta_{3} + \dots + \Lambda_{NN-1}^P\eta_{N-1} \bigg] -\gamma_N + \mathbf{J}_\mu f(\bar{x},\bar{\mu})w_{N}, \label{MSf 3} \\ \nonumber
    \dot \gamma_{N} &= \mathbf{J}_xg(0)\bigg[\Lambda_{N2}^I\eta_{2} + \Lambda_{N3}^I\eta_{3} + \dots + \Lambda_{NN-1}^I\eta_{N-1} \bigg],
\end{align}
\end{subequations}
where
\[
w_{i}= 
    \sum_{k=1}^N T^{-1}_{ik}\delta \mu_k  \qquad i=1,\ldots,N.
\]
Equation \eqref{MSf syn} represents the linearized dynamics tangential to the synchronization manifold and, consequently, does not influence the manifold's local stability. Instead, the stability of the synchronous state is determined by the evolution of transverse perturbations. Specifically, the synchronization manifold is locally asymptotically stable if the $\eta_i$ converge to the zero for all $i = 2, \dots, N$. Given that $\mathbf{J}_\mu f(\bar{x}, \bar{\mu})$ is constant and the spectrum of the integral layer Laplacian, $\mathrm{spec}(\widetilde{L}^I)$, possesses a simple zero eigenvalue from Assumption \ref{ass:1}, the equilibrium points of the variational system \eqref{MSf 1} are:

\begin{equation*}
\begin{bmatrix}
\eta_{2}^* \\
\gamma_{2}^*
\end{bmatrix}
=
\begin{bmatrix}
0_n \\
\mathbf{J}_\mu f(\bar{x},\bar{\mu}) w_{2}
\end{bmatrix},
\end{equation*}
with its stability determined by the maximum Lyapunov exponent (MLE) of the dynamic matrix
\begin{equation}
\mathcal{M}_2 = 
\begin{bmatrix}
\mathbf{J}_x f(\bar{x},\bar{\mu}) - \lambda_2^P \mathbf{J}_x g(0) & -\mathbf{I}_n \\
\lambda_2^I \mathbf{J}_x g(0) & \mathbf{0}_n
\end{bmatrix}.
\end{equation}
Since the MLE for \eqref{MSf 1} is negative, then the transverse perturbations $\eta_{2}$ decay exponentially to zero and $\gamma_2$ to $\mathbf{J}_\mu f(\bar x,\bar \mu)w_2$ as $t \to \infty$. Given that $\|\mathbf{J}_x g(0)\|$ is bounded and  \eqref{MSf 1} has a negative MLE, the stability of the subsequent layer \eqref{MSf 2} is determined by the MLE of the dynamics matrix $ \mathcal{M}_3$, which is negative by hypothesis, 
yielding $\eta_{3} \to 0_n$ and $\gamma_{3} \to \mathbf{J}_\mu f(\bar{x},\bar{\mu}) w_{3}$ in the limit $t \to \infty$. By induction, this argument establishes that the transverse perturbations $\eta_i$, $i=2,\ldots,N$ converge to $0_n$ (and $\gamma= [T^{-1}\otimes\mathbf{J}_\mu f(\bar x, \bar \mu)]\delta \mu$ since $\delta u(0)= [-\mathbf{I}_N \otimes\mathbf{J}_\mu f(\bar x, \bar \mu
)]\delta \mu$) since the MLE of the dynamic matrix $\mathcal{M}_i$ is negative for all $\lambda_i$ with $i=2,\dots,N$. This completes the proof.

\end{pf}

A particular instance of transformation $T$ that fulfills Assumption \ref{ass:1} is, when it exists, the transformation that simultaneously Jordanize matrices $\widetilde{L}^P$ and $\widetilde{L}^I$. Accordingly, we can give the following corollary:

\begin{corollary}
\label{cor2.1}
If $\widetilde{L}^I$ has a simple zero eigenvalue, and $\widetilde{L}^P$ and $\widetilde{L}^I$ are simultaneously Jordanizable, $\|\mathbf{J}_x g(0)\|$ is bounded,  $\mathbf{J}_\mu f(\bar x, \bar \mu)$ is constant, $\delta u(0)= [-\mathbf{I}_N \otimes\mathbf{J}_\mu f(\bar x, \bar \mu
)]\delta \mu$,
and the maximum Lyapunov exponent (MLE) of
\begin{equation}
\mathcal{M}_i = 
\begin{bmatrix}
\mathbf{J}_x f(\bar{x},\bar{\mu}) - \lambda^P_i \mathbf{J}_x g(0) & -\mathbf{I}_n \\
\lambda^I_i \mathbf{J}_x g(0) & \mathbf{0}_n
\end{bmatrix}.
\end{equation}
is negative for all $\lambda^P_i$ and $\lambda^I_i$ where $i=2, \dots, N$, then the synchronization manifold under dynamic diffusive coupling is locally transversely stable.
\end{corollary}

\begin{pf}
The proof trivially follows noting that Assumption \ref{ass:1} is fulfilled and applying Theorem \ref{Theorem 2}.    
\end{pf}

Notice that, under the assumptions of Theorem \ref{Theorem 2} and Corollary \ref{cor2.1}, the transverse stability of the synchronization manifold depends on the MLE associated to the blocks $\mathcal M_2,\ldots,\mathcal M_N$. Therefore, the following master stability equation can be used to check transversal stability:
\begin{align} \label{MSF}
\begin{bmatrix}
    \dot{\eta} \\
    \dot{\gamma}
\end{bmatrix} = 
\begin{bmatrix}
    \mathbf{J}_x f(\bar{x}, \bar{\mu}) - \varphi \mathbf{J}_x g(0) & -\mathbf{I}_n \\
    \varrho\mathbf{J}_x g(0) & \mathbf{0}_n
\end{bmatrix}
\begin{bmatrix}
    \eta \\
    \gamma
\end{bmatrix},
\end{align}
where $\varphi, \varrho \in \mathbb{C}$ are complex parameters. The Master Stability Function (MSF) for the network is defined as the maximum Lyapunov exponent (MLE) associated with \eqref{MSF} as a function of $\varphi$ and $\varrho$, i.e., $\mathrm{MSF(\varphi,\varrho)}$. To check transversal stability of a given network topology, one would then need to check that
\[
\max_{i=2,\ldots,N}\mathrm{MSF}(\lambda_i^P,\lambda_i^I)<0,
\]
which is equivalent to require the negativity of the MLE for $\mathcal M_2,\ldots,\mathcal M_N$, as reported in the statement of Theorem \ref{Theorem 2} and Corollary \ref{cor2.1}.

In practice, one first computes the $\mathrm{MSF}$ over the relevant $(\varphi,\varrho)$ domain and then checks whether all pairs $(\lambda_i^P,
\lambda_i^I)$ lie inside the region where the MSF is negative.
Additionally, note that since the eigenvalues of the Laplacian matrices $\widetilde L^P$ and $\widetilde L^I$ may also be negative, the MSF should also be evaluated for negative values of $\Re(\varphi)$ and $\Re(\varrho)$ \cite{della_rossa_emergence_2023}.


\begin{rmk} \label{Remark 2}
All the analysis performed in this section relies on the presence of the integral action, therefore the MSF is not meaningful for $\varrho=0$ (that corresponds to no integral action). In the absence of an integral action, the synchronization manifold is no longer invariant, and a bound on the synchronization error can be found, see e.g. \cite{rizzello_pinning_2024} for directed hypergraphs, or \cite{nishikawa_synchronization_2006} for digraphs. 
\end{rmk}

When Assumption \ref{ass:1} does not hold, there is still a relevant instance in which the MSF can be used to evaluate the transverse stability of the synchronization manifold. This is the case when the integral layer is a complete $k$-uniform hypergraph, with each hyperedge sharing the same coupling strength $\sigma_I$, and the other layer being any directed hypergraph. In this scenario, it is possible to state the following theorem:

\begin{theorem} \label{Theorem 3}
If the integral layer is a complete $k$-uniform hypergraph with $(\sigma_{e})_I=\sigma_I$ for all $e\in\mathcal E_I$,  $\|\mathbf{J}_x g(0)\|$ is bounded, $ \mathbf{J}_\mu f(\bar x, \bar \mu)$ is constant, $\delta u(0) = [-I_N\otimes \mathbf{J}_\mu f(\bar x, \bar \mu)]\delta\mu$,
$\max_{i=2,\ldots,N}\mathrm{MSF}(\lambda_i^P,\sigma_I N)<0$, then the synchronization manifold under dynamic diffusive coupling is locally transversely stable.

\end{theorem}

\begin{pf}

Let $V$  be  an invertible transformation matrix which diagonalizes $\widetilde{L}^I$  such that $\widetilde{L}^P = V\Lambda^PV^{-1}$ and $\widetilde{L}^I =  V\Lambda^IV^{-1}$ with

\begin{equation}
    \Lambda^P =
    \left[ \begin{array}{ccc|c}
         & & &  \\
         & \Gamma & & 0_{N-1}\\
         & & & \\ \hline 
        \rule{0pt}{15pt} \Lambda^P_{N1} & \dots & \Lambda^P_{NN-1} & 0
    \end{array}
    \right],
\quad
    \Lambda^I =
    \left[ \begin{array}{ccc|c}
         & & &  \\
         & \sigma_I N I_{N-1} & & 0_{N-1}\\
         & & & \\ \hline 
        \rule{0pt}{15pt}  & 0^T_{N-1} &  & 0
    \end{array}
    \right],
\end{equation}
%
%
where $\Gamma\in\mathbb{R}^{(N-1)\times (N-1)}$, and the last column of both $\Lambda^P$ and $\Lambda^I$ are equal to $0_N$ as the Laplacian matrices are zero row-sum and have $1_N$ as one of their eigenvectors. Using the transformation $\varsigma = (V^{-1} \otimes I_n)\delta x$, we could write \eqref{eq:deltax_dot} and \eqref{eq:deltau_dot} as follows

\begin{subequations} 
\begin{align}
    \dot \varsigma' &= \bigg[I_{N-1} \otimes\mathbf{J}_xf(\bar{x},\bar{\mu})- \Gamma\otimes \mathbf{J}_xg(0)\bigg]\varsigma'+\bigg[\bar V\otimes\mathbf{J}_\mu  f(\bar{x},\bar{\mu})\bigg]\delta \mu - \varepsilon'  \label{eq1}\\
    \dot \varepsilon' &=  \sigma_I N\bigg[I_{N-1}\otimes\mathbf{J}_xg(0)\bigg]\varsigma' \label{eq2}\\
    \dot \varsigma_N &= \mathbf{J}_xf(\bar{x},\bar{\mu})\varsigma_N - \mathbf{J}_xg(0)\sum_{j=1}^{N-1}\Lambda_{Nj}\varsigma_j \\
    \dot \varepsilon_N &=  0_n
\end{align} 
\end{subequations}
where $\varsigma' = [\varsigma_1;\varsigma_2;\dots ;\varsigma_{N-1}]$ is the perturbations transverse to the synchronization manifold and $\varepsilon' = [\varepsilon_1;\varepsilon_2;\dots ;\varepsilon_{N-1}]$ with $\varepsilon = -(V^{-1} \otimes I_n)\delta u$. 

Next, let us denote $\bar{V}$ as the matrix obtained by extracting the first $N-1$ rows of $V^{-1}$. Moreover,
let $Q \in \mathbb{R}^{N-1 \times N-1}$  be an invertible transformation matrix which puts $\Gamma$ in its Jordan canonical form $\Lambda^\Gamma=Q^{-1}\Gamma Q$. Hence using the transformation  $\eta' = (Q^{-1} \otimes I_n)\varsigma'$, we could write \eqref{eq1} and \eqref{eq2} as follows:
\begin{subequations}   \label{error dyamics thr3}
\begin{align}
    \dot \eta' &= \bigg[I_{N-1} \otimes\mathbf{J}_xf(\bar{x},\bar{\mu})- \Lambda^\Gamma\otimes \mathbf{J}_xg(0)\bigg]\eta'+\bigg[ Q^{-1}\bar V\otimes\mathbf{J}_\mu  f(\bar{x},\bar{\mu})\bigg]\delta \mu - \gamma'  \\
    \dot \gamma' &=  \sigma_I N\bigg[I_{N-1}\otimes\mathbf{J}_xg(0)\bigg]\eta' 
\end{align} 
\end{subequations}
where $\gamma' = (Q^{-1} \otimes I_n)\varepsilon'$. Similar to the case discussed in Theorem \ref{Theorem 2}, the transverse stability of the synchronization manifold then follows from the hypothesis that $\max_{i=2,\ldots,N}\mathrm{MSF}(\lambda_i^P,\sigma_I N)<0$.



\end{pf}

\section{Numerical Applications} \label{sec:Numerical Applications}
In this section, we 1) use the $\mathrm{MSF}$ approach to make some salient observations on synchronizability and validate the theoretical findings of Section \ref{Results}, using as a testbed a network of coupled Lorenz oscillators, and 2) we apply the results on a control problem in opinion dynamics.
\subsection{Synchronization of Lorenz systems coupled on directed hypergraphs in the presence of heterogeneous constant disturbances}
We investigate the problem of achieving complete synchronization in a network of $N$ Lorenz systems subject to constant heterogeneous disturbances $\mu_i$. The vector field that describes the individual dynamics of the $i$th Lorenz system is
\begin{equation} \label{eq: lorenz system}
    f(x_i,\mu_i) = \begin{bmatrix}
        \varsigma(x_{i2} - x_{i1}) +\mu_i\\
        x_{i1}(\rho - x_{i3}) - x_{i2} \\
        x_{i1}x_{i2} - \vartheta x_{i3}
    \end{bmatrix},
\end{equation}
where $\varsigma=10$, $\rho=28$, $\vartheta=8/3$, so that, in the absence of coupling and for $\mu_i$ in [0,1], the dynamics would have a chaotic attractor \cite{lorenz_deterministic_1963}; the disturbance $\mu_i$ is randomly extracted from a uniform distribution in $[0,1]$. Denoting $x_a \in \mathbb{R}^3$ as a point belonging to the attractor of $f(x,\mu)$ with $f$ defined by $\eqref{eq: lorenz system}$ and $\mu = \frac{1}{N}\sum_{j=1}^N \mu_i$, we set the initial conditions for the simulations as $x_i(0) = x_a + p_i1_3$ where $p_i$ is randomly extracted from a uniform distribution in $[0,10^{-6}]$. In all simulations, we consider the coupling gain of the proportional layer to be identical, that is, $(\sigma_e)_P=\sigma_P$ for all $e\in\mathcal E_P$. Similarly, for the integral layer, we consider $(\sigma_e)_I=\sigma_I$ for all $e\in\mathcal E_I$. As for the coupling function $g$, in our simulations we made two alternative choices, i) $g(z)=[z_{1};0;0]$, and ii) $g(z)=[\tanh(z_{1});0;0]$, which share the same Jacobian, so that the network would have the same MSF.

\begin{figure}
    \centering
    \includegraphics[width=0.6\linewidth]{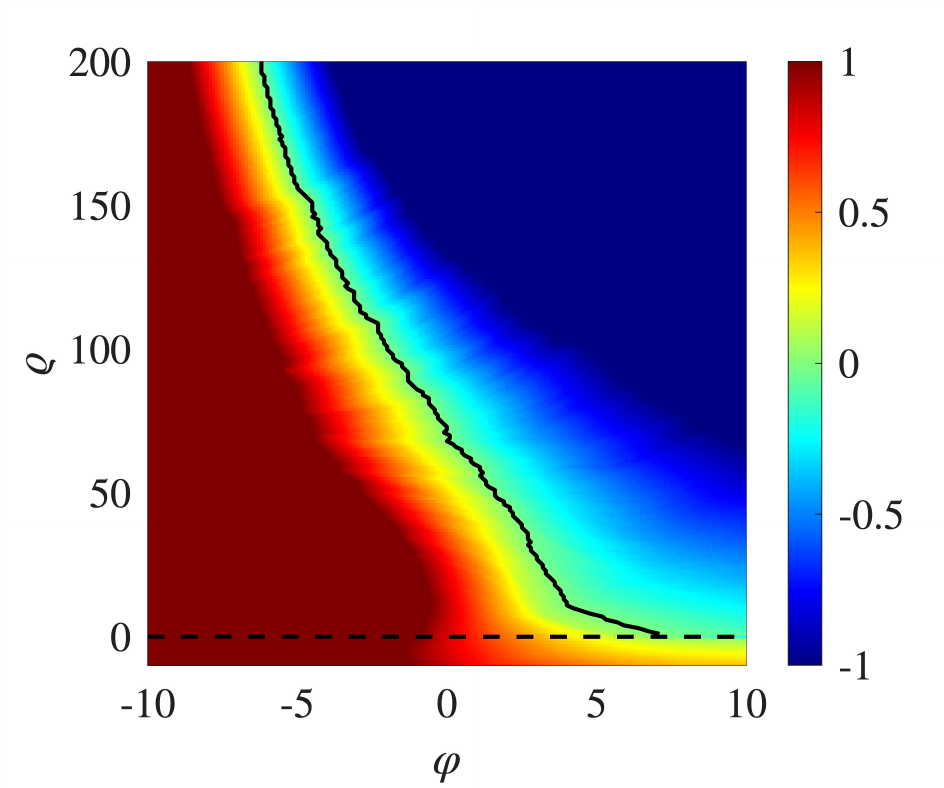}
    \caption{Master Stability Function for a network of chaotic Lorenz systems with $\mathbf{J}_xg(0)=\mathrm{diag}([1,0,0]^\mathrm{T})$, whose two layers have a real spectrum, which implies $\varphi,\varrho\in\mathbb R$. The solid black curve denotes the zero-contour of the MSF, delimiting the regions with local transverse stability. The dashed black line at $\varrho=0$ reminds that the analysis does not apply in the absence of an integral action, whereby the synchronization manifold would not be invariant.}
    \label{Fig.2}
\end{figure}

Following \cite{burbano_lombana_synchronization_2016}, to quantify the degree of synchronization in the network, we introduce the network disagreement $d$, defined as
\begin{equation}
    d(t)=\|[\delta x_1;\ldots;\delta x_N]\|
\end{equation}
and compute its time-average after removing the transient response, that is,
\begin{equation}
\langle d\rangle=\frac{1}{1000}\int_{T-1000}^Td(t)\mathrm{d}t,
\end{equation}
where $T$ is the simulation time (selected as $T=6\times 10^5$ time units in our simulations and integrated using $\mathrm{ode45}$).

In what follows, we first report an observation on the role played by the integral layer on network synchronizability, to then validate the stability results derived in Section \ref{Results}.
\subsubsection{The integral action alone can lead to complete synchronization}
Different from the case of digraphs, where Laplacians only have nonnegative real-part eigenvalues, directed hypergraphs can be instead associated to Laplacians with one or more negative real-part eigenvalues, which constitute a further challenge for synchronizability.
Here, we illustrate how the integral action can effectively drive the network to complete synchronization even in the presence of negative real-part eigenvalues in the Laplacian $\widetilde L_P$ associated to the proportional layer. First, we start by noting that, as the coupling strengths are homogeneous within each layer, the values of $\varphi$ and $\varrho$ at which we need to evaluate the MSF for a given network, would be
$(\varphi=\sigma_P\lambda_i^P,\varrho=\sigma_I\lambda_i^I)$, for $i=2,\ldots,N$. This means that, if $\lambda_i^P$ has negative real-part for some $i$, then in order to attain complete synchronization, $\mathrm{MSF}(\varphi,\varrho)$ should be negative even for some negative real-part $\varphi$.

To illustrate the potential effectiveness of the integral action in dealing with the case of negative real-part eigenvalues $\lambda_i^P$, we start by plotting the $\mathrm{MSF}$ for network topologies characterized by real eigenvalues, so that the $\mathrm{MSF}$ can be effectively depicted by a colormap in 2-D, see Fig.~\ref{Fig.2}. Notice that the $\mathrm{MSF}$ is negative also for negative values of $\varphi$, thus implying that complete synchronization is feasible also when there are one or more eigenvalues $\lambda_i^P$ that are negative.

In Fig. \ref{Fig.6}, we show a concrete example of a $N=100$ network in which the proportional layer has a negative eigenvalue, equal to $-0.216$. Specifically, the topology is generated as an ER-like directed hypergraph with only triadic interactions, by following the procedure used in \cite{rizzello_pinning_2024}. The integral layer is selected as a complete graph, so that Theorem \ref{Theorem 3} does apply, and the $\mathrm{MSF}$ can be used to assess local transverse stability. Fig. \ref{Fig.6.a} shows a perfect matching between the theoretical prediction and the actual synchronizability region, whereby the contour corresponding to the pairs $(\sigma_P,\sigma_I)$ such that $\max_{i=2,\ldots,N}\mathrm{MSF}(\sigma_P\lambda_i^P,\sigma_I\lambda_i^I)=0$  correctly identifies the boundaries of the stability region. Moreover, it confirms that the integral action is able to counteract the destabilizing effect of the negative eigenvalue of $\widetilde L^P$, with increasing values of  $\sigma_I$ required as the coupling strength of the proportional layer $\sigma_P$ increases. Panels (b) and (c) report two representative instances in which the strength of the integral action is insufficient or sufficient to guarantee convergence of the disagreement to zero, respectively.

\begin{figure*} 
    \centering
    \begin{subfigure}[b]{0.32\textwidth}
        \centering
        \includegraphics[width=0.95\textwidth]{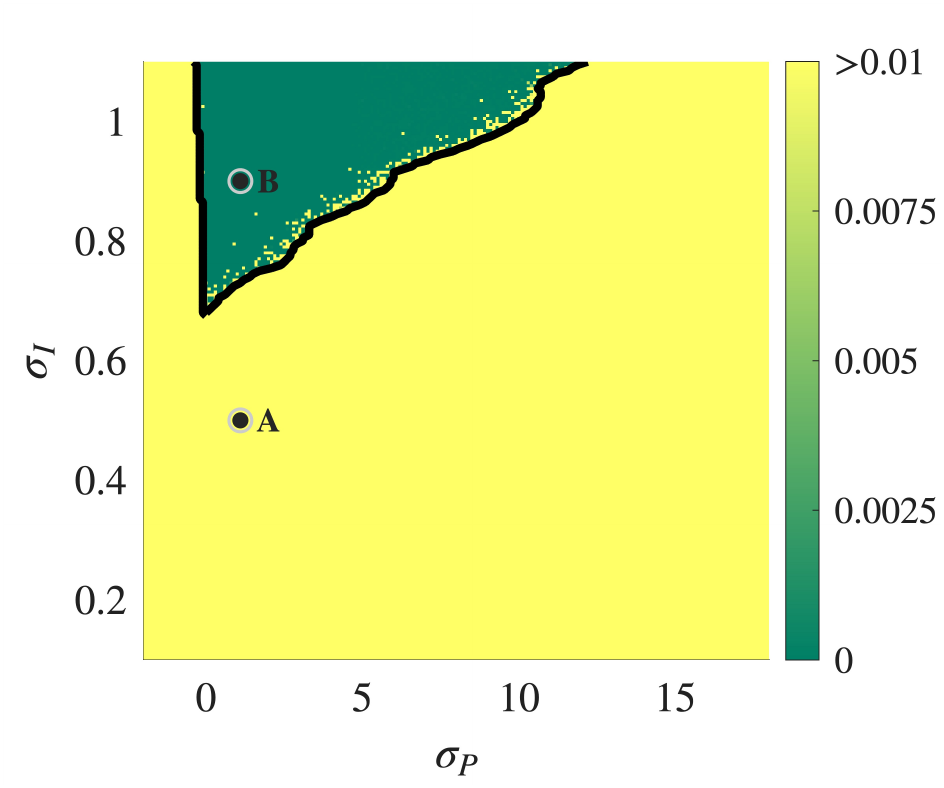}
        \caption{}
        \label{Fig.6.a}
    \end{subfigure}
    \hfill
    \begin{subfigure}[b]{0.32\textwidth}
        \centering
        \includegraphics[width=0.95\textwidth]{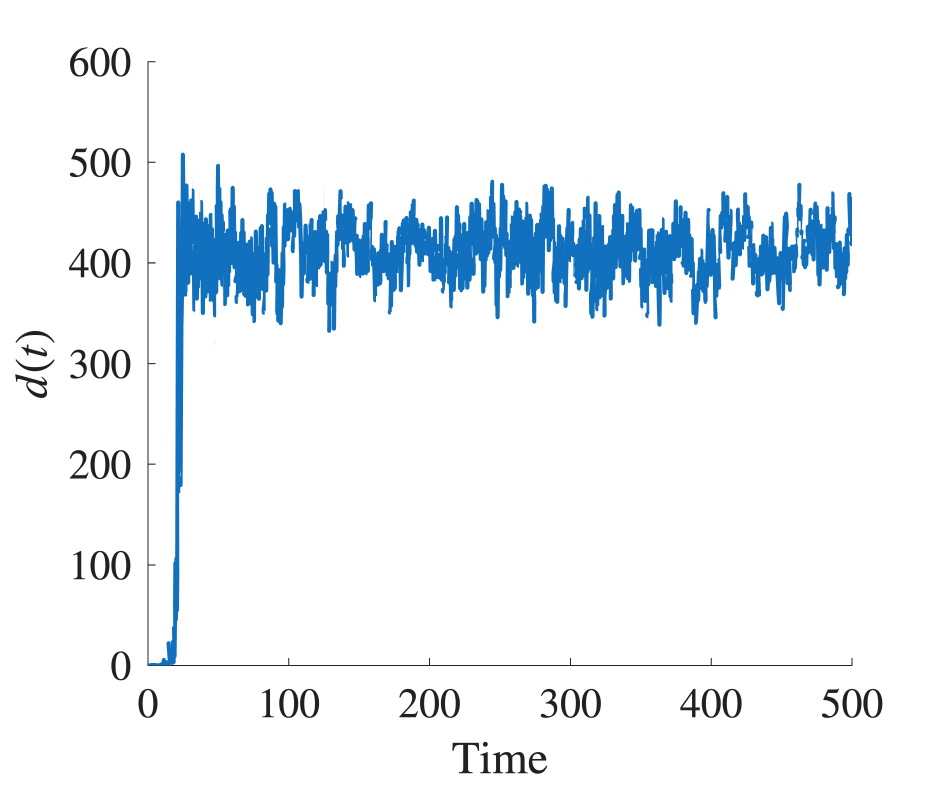}
        \caption{}
        \label{Fig.6.b}
    \end{subfigure}
    \begin{subfigure}[b]{0.32\textwidth}
        \centering
        \includegraphics[width=0.95\textwidth]{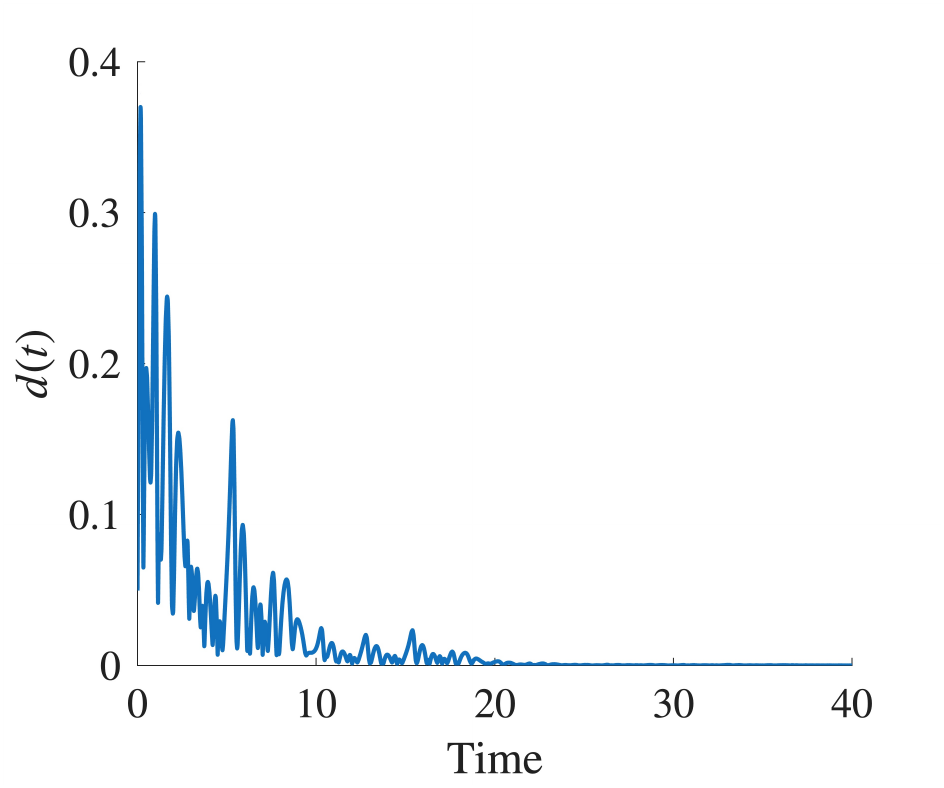}
        \caption{}
        \label{Fig.6.c}
    \end{subfigure}
    \caption{Network of $N=100$ chaotic Lorenz systems with coupling function $g(z)=[\tanh(z_1);0;0]$. The proportional layer is an ER-like directed hypergraph with only triadic interactions and average degree equal to 5, whereas the integral layer is a complete graph. 
    Panel (a) reports a colormap for the time-averaged disagreement $\langle d\rangle$ as a function of the proportional and integral coupling strength $(\sigma_P,\sigma_I)$. The solid black line identifies the boundary of the stability region obtained by computing the values of $\sigma_P$ and $\sigma_I$ for which $\max_{i=2,\ldots,N}\mathrm{MSF}(\sigma_P\lambda_i^P,\sigma_I\lambda_i^I)=0$. Panels (b) and (c) reports the time evolution of the disagreement $d(t)$ for the pairs $(\sigma_P,\sigma_I)$ corresponding to points $\mathrm{A}=(1.1,0.5)$ and $\mathrm{B}=(1.1,0.9)$ in panel (a), respectively. }
    \label{Fig.6}
\end{figure*}

Finally, to get a further insight on the shape of the $\mathrm{MSF}$, we now consider the case in which $\varphi$ and $\varrho$ are complex numbers, which is necessary to analyze instances where the two layers have Laplacians with complex spectra. In this case, the $\mathrm{MSF}$ would be function of 4 parameters, and therefore to visualize it, in Fig.~\ref{Fig.3} we fix the value of $\varphi$, thus obtaining 2-dimensional colormaps of the $\mathrm{MSF}$. In particular, in Fig.~\ref{Fig.3.a}, we show that, in the absence of a proportional coupling layer ($\varphi=0$), one can always find a minimal $\Re(\varrho)$ beyond which complete synchronization is attained, and this threshold increases with $\Im(\varrho)$, that is, it increases in the presence of complex eigenvalues in the integral layer. A similar behavior is also observed in the presence of a proportional control action, see Fig. \ref{Fig.3.b}. 
\begin{figure*} 
    \centering
    \begin{subfigure}[b]{0.48\textwidth}
        \centering
        \includegraphics[width=0.7\textwidth]{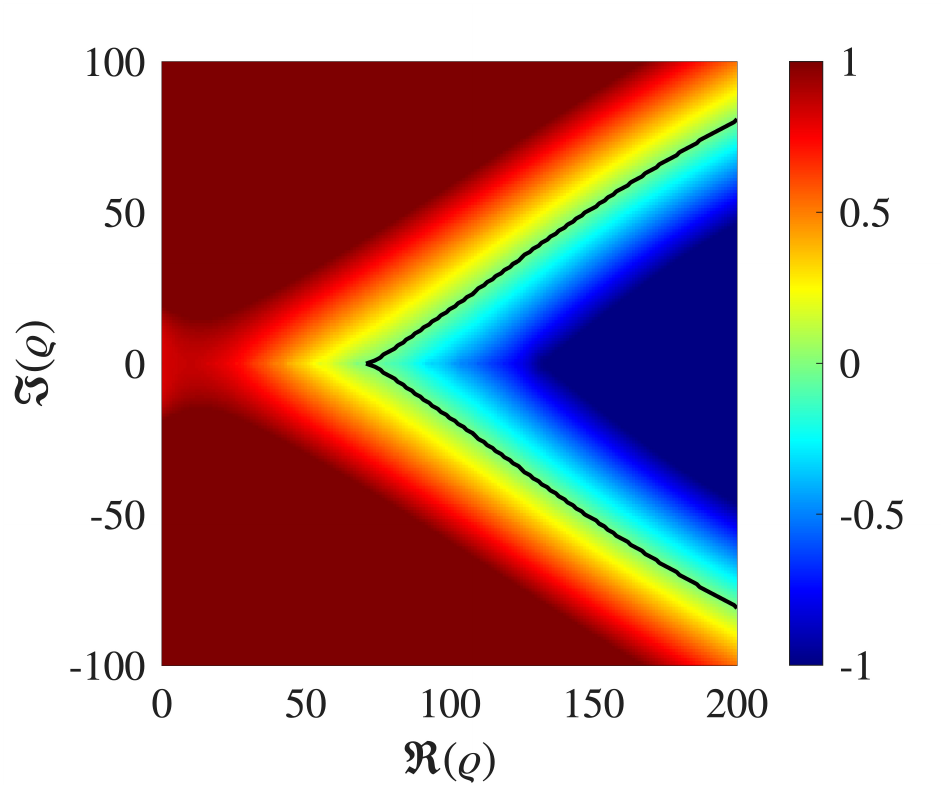}
        \caption{$\Re(\varphi)=\Im(\varphi)=0$}
        \label{Fig.3.a}
    \end{subfigure}
    \hfill
    \begin{subfigure}[b]{0.48\textwidth}
        \centering
        \includegraphics[width=0.7\textwidth]{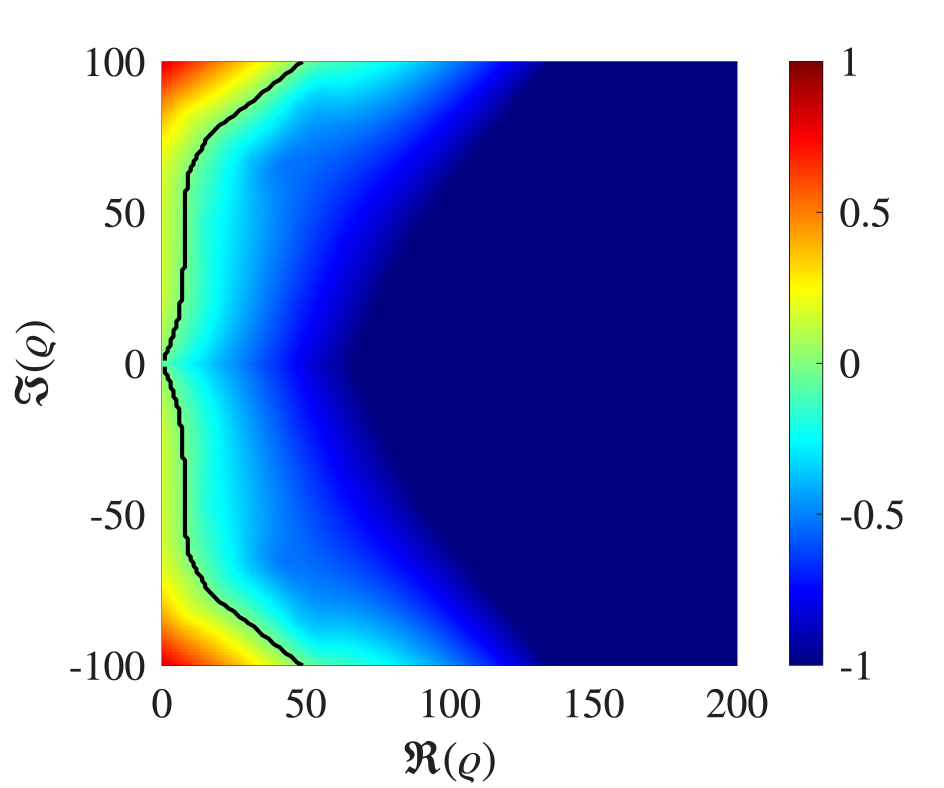}
        \caption{$\Re(\varphi)=10$ and $\Im(\varphi)=0$}
        \label{Fig.3.b}
    \end{subfigure}
    \caption{Master Stability Function $\mathrm{MSF}(\varphi, \varrho)$ for a network of chaotic Lorenz system with $\mathbf{J}_xg(0)=\mathrm{diag}([1,0,0]^\mathrm{T})$ as a function of $\varrho\in\mathbb C$, with $\varphi=0$ (panel a) and $\varphi=10$ (panel b). The solid black curve denotes the zero-contour of the MSF, delimiting the regions with local transverse stability.}
    \label{Fig.3}
\end{figure*}

\subsubsection{Validation of Theorem \ref{Theorem 2}}

Here, we start by considering the case in which Assumption \ref{ass:1} is fulfilled since there exists a transformation $T$ that puts the Laplacian in the form reported in \eqref{eq:lambdas}. Specifically, for the proportional and integral layer we consider two directed  hypergraphs of $N=6$ nodes, whose Laplacians $\widetilde{L}^P$ and $\widetilde{L}^I$, transformation $T$, and matrices $\Lambda_P$ and $\Lambda_I$ are reported in Fig. \ref{Fig.4.a}.
Fig. \ref{Fig.4.b} reports the time-averaged disagreement $\langle d\rangle$ showing how the transition to synchronization is accurately predicted through the use of the $\mathrm{MSF}$. The time-evolution of the disagreement is depicted for two representative pairs $(\sigma_P,\sigma_I)$ in panels (c) and (d). Similar results are obtained for larger ($N=100$) networks, as illustrated in Figure \ref{Fig.5}, where two identical ER-like directed hypergraphs with only triadic interactions are considered for the proportional and integral layers.

%
%
%

\begin{figure*} 
    \centering
    \begin{subfigure}[b]{0.95\textwidth}
        \centering
        \includegraphics[width=0.95\textwidth]{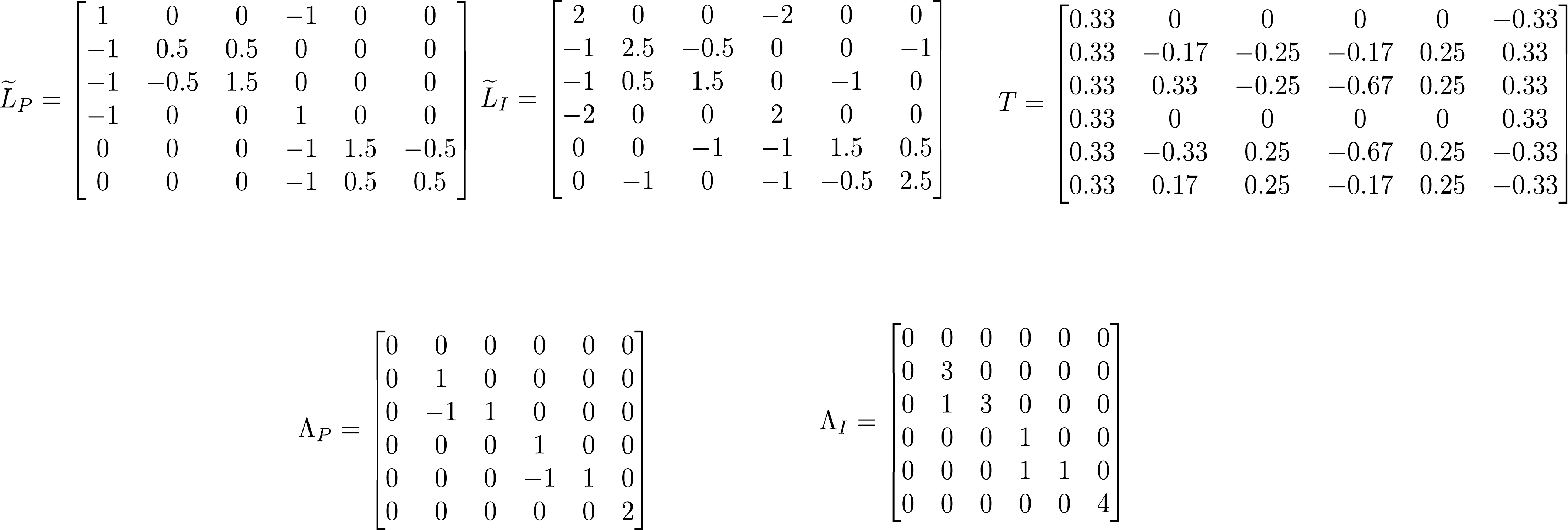}
        \caption{}
        \label{Fig.4.a}
    \end{subfigure}
    \hfill
    \begin{subfigure}[b]{0.32\textwidth}
        \centering
        \includegraphics[width=0.95\textwidth]{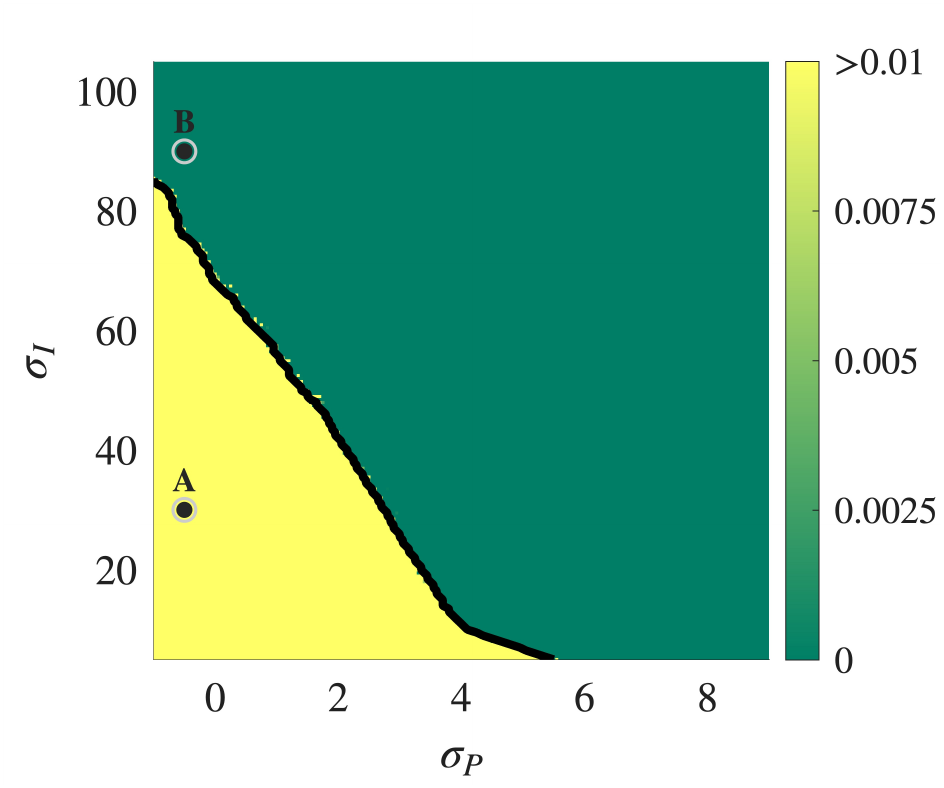}
        \caption{}
        \label{Fig.4.b}
    \end{subfigure}
    \hfill
    \begin{subfigure}[b]{0.32\textwidth}
        \centering
        \includegraphics[width=0.95\textwidth]{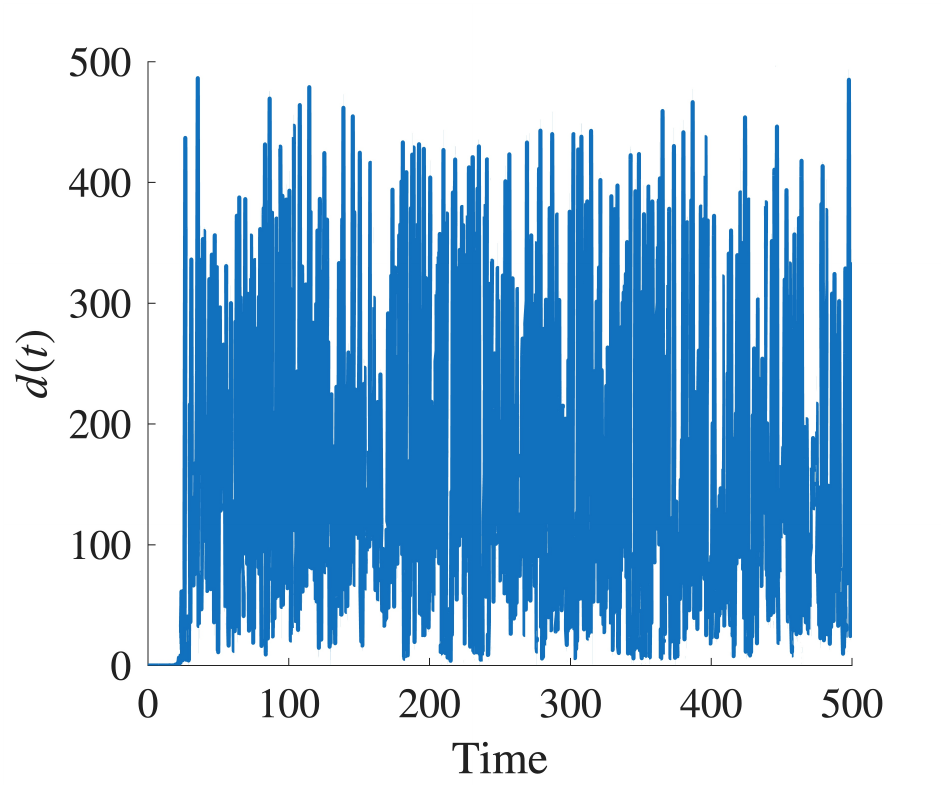}
        \caption{}
        \label{Fig.4.c}
    \end{subfigure}
    \begin{subfigure}[b]{0.32\textwidth}
        \centering
        \includegraphics[width=0.95\textwidth]{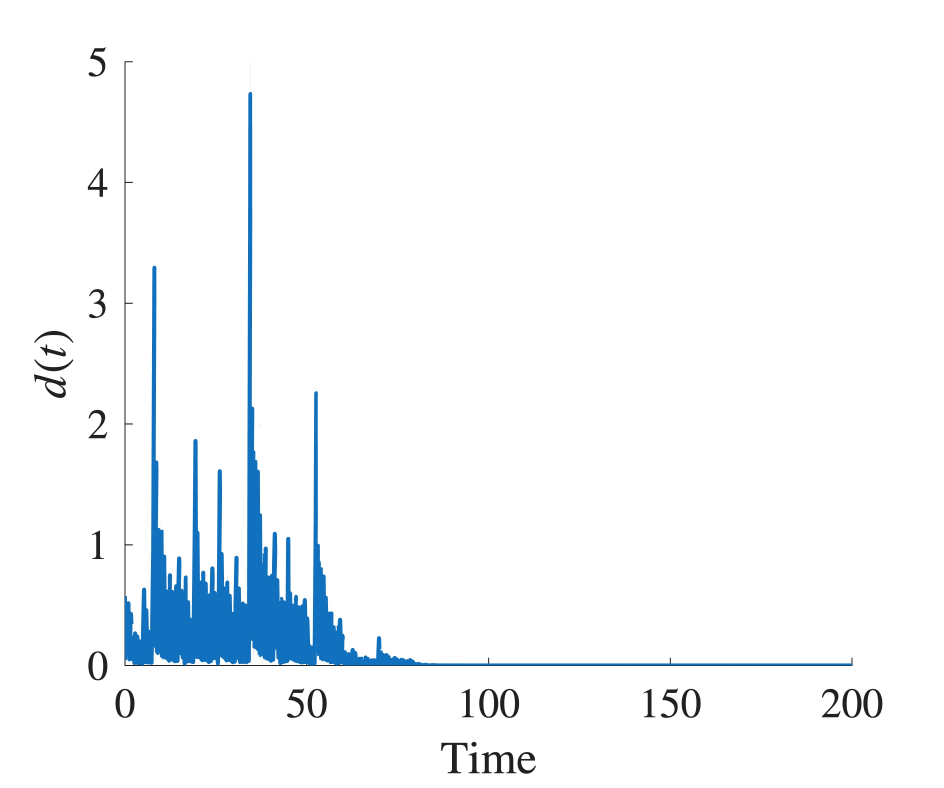}
        \caption{}
        \label{Fig.4.d}
    \end{subfigure}
    \caption{Network of chaotic Lorenz systems with coupling function $g(z)=[z_1;0;0]$. Panel (a) reports the Laplacians $\widetilde L_P$, $\widetilde L_I$, the transformation $T$, and matrices $\Lambda_P$ and $\Lambda_I$.
    Panel (b) reports a colormap for the time-averaged disagreement $\langle d\rangle$ as a function of the proportional and integral coupling strength $(\sigma_P,\sigma_I)$. The solid black line identifies the boundary of the stability region obtained by computing the values of $\sigma_P$ and $\sigma_I$ for which $\max_{i=2,\ldots,N}\mathrm{MSF}(\sigma_P\lambda_i^P,\sigma_I\lambda_i^I)=0$. Panels (c) and (d) reports the time evolution of the disagreement $d(t)$ for the pairs $(\sigma_P,\sigma_I)$ corresponding to points $\mathrm{A}=(-0.5,30)$ and $\mathrm{B}=(-0.5,90)$ in panel (a), respectively.}
    \label{Fig.4}
\end{figure*}

\begin{figure*} 
    \centering
    \begin{subfigure}[b]{0.32\textwidth}
        \centering
        \includegraphics[width=0.95\textwidth]{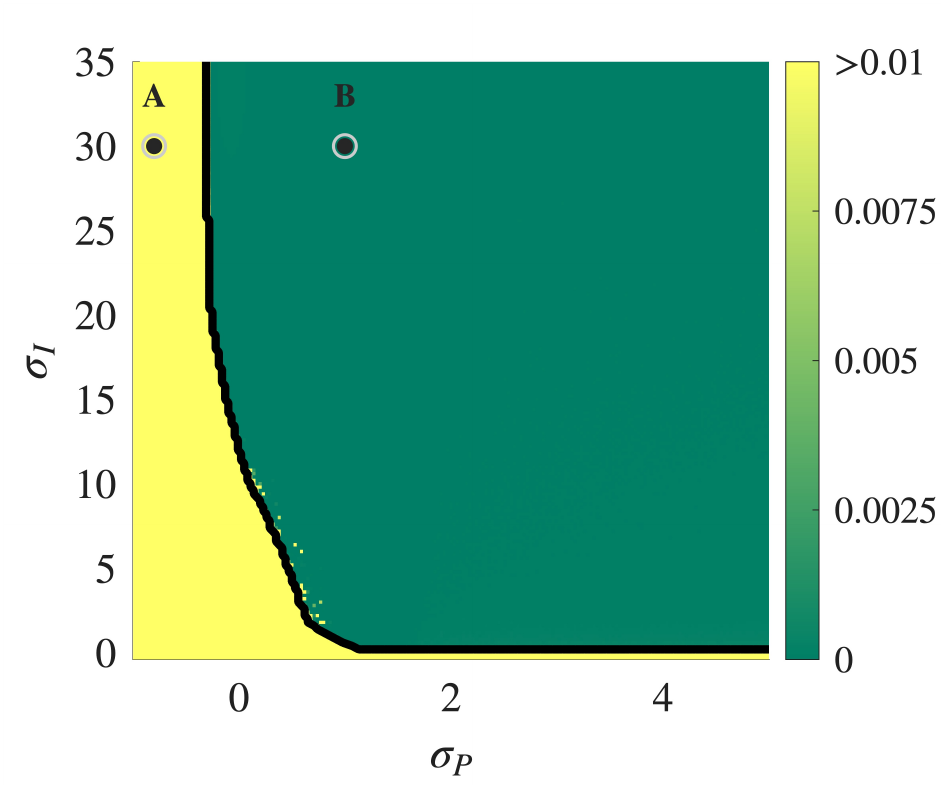}
        \caption{}
        \label{Fig.5.a}
    \end{subfigure}
    \hfill
    \begin{subfigure}[b]{0.32\textwidth}
        \centering
        \includegraphics[width=0.95\textwidth]{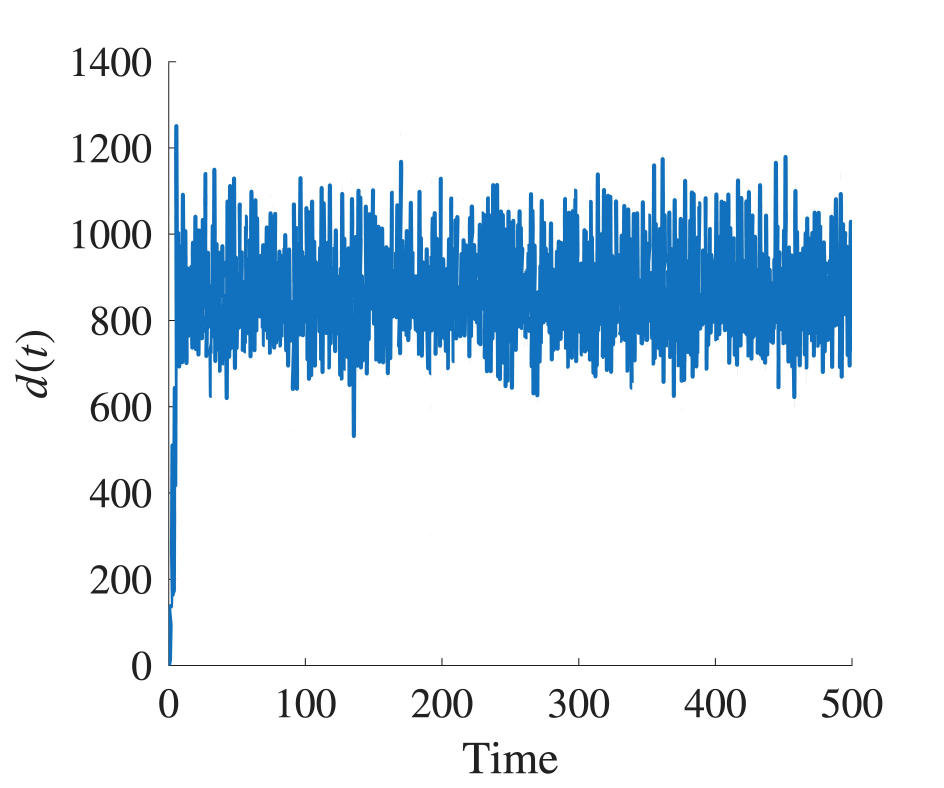}
        \caption{}
        \label{Fig.5.b}
    \end{subfigure}
    \begin{subfigure}[b]{0.32\textwidth}
        \centering
        \includegraphics[width=0.95\textwidth]{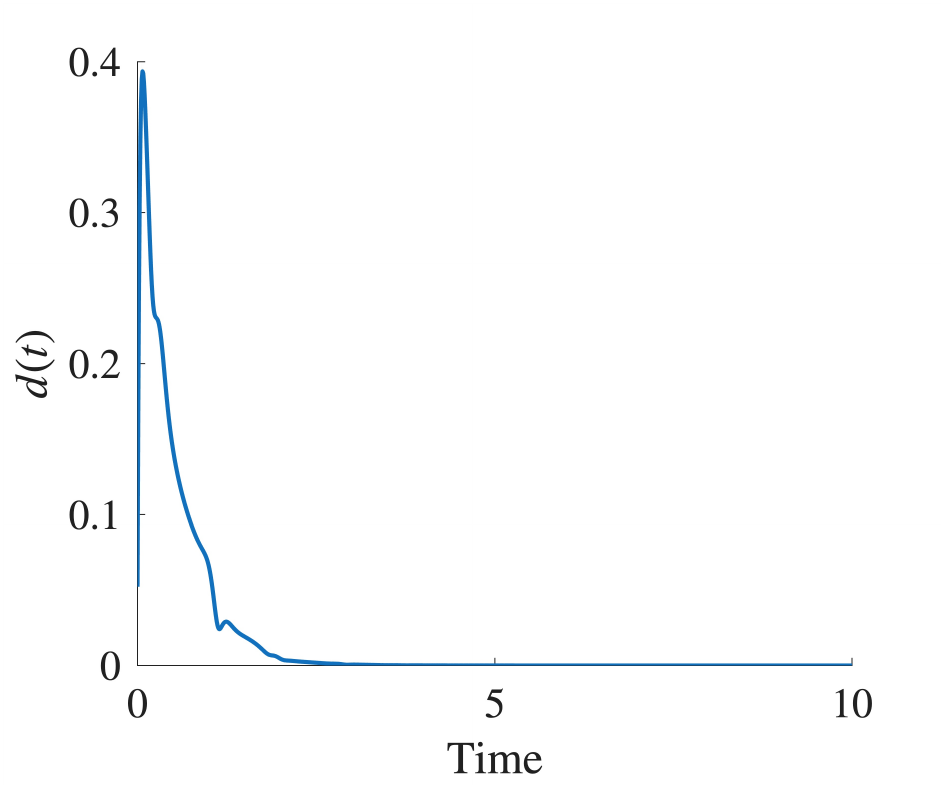}
        \caption{}
        \label{Fig.5.c}
    \end{subfigure}
    \caption{Network of $N=100$ chaotic Lorenz systems with coupling function $g(z)=[\tanh(z_1);0;0]$. The proportional and integral layer are coincident, and correspond to  an ER-like directed hypergraph with only triadic interactions and average degree equal to 20. 
    Panel (a) reports a colormap for the time-averaged disagreement $\langle d\rangle$ as a function of the proportional and integral coupling strength $(\sigma_P,\sigma_I)$. The solid black line identifies the boundary of the stability region obtained by computing the values of $\sigma_P$ and $\sigma_I$ for which $\max_{i=2,\ldots,N}\mathrm{MSF}(\sigma_P\lambda_i^P,\sigma_I\lambda_i^I)=0$. Panels (b) and (c) reports the time evolution of the disagreement $d(t)$ for the pairs $(\sigma_P,\sigma_I)$ corresponding to points $\mathrm{A}=(-0.8,30)$ and $\mathrm{B}=(1,30)$ in panel (a), respectively.}
    \label{Fig.5}
\end{figure*}

\subsection{An opinion dynamics application}
To demonstrate the advantage of integral action, we consider a numerical case study adapted from \cite{rizzello_pinning_2024}, involving a real-world social network topology comprising $N=150$ agents, which we use as the directed hypergraph describing both for the proportional and the integral coupling layer. The scenario models a binary-choice consensus process where individuals are influenced by a central opinion leader, hereafter referred to as the pinner \cite{de_lellis_pinning_2023}. Unlike the purely diffusive coupling explored in previous literature, which can only bound the opinion mismatch with the leader, we implement our dynamic PI-coupling scheme to manage nodal heterogeneities and attain complete leader-follower consensus.

The state $x_i \in \mathbb{R}$ represents the opinion of agent $i$, where $x_i = 0$ signifies a neutral stance, and $x_i > 0$ or $x_i < 0$ indicate an inclination toward one of the two competing options. The uncoupled dynamics of each agent are governed by the following nonlinear function:
\begin{equation}
\label{eq:opinion}
    f(x_i, \mu_i) = -3x_i + \mu_i \tanh(x_i).
\end{equation}
For $\mu_i > 3$, the origin becomes unstable, giving rise to two symmetric stable equilibria at $\pm x_i^*$. In this context, each hyperedge with cardinality $\mathcal T(e)+\mathcal H(e)\ge 3$ represents a higher-order group interaction rather than a simple pairwise link \cite{stehle_high-resolution_2011}. We assume the proportional and integral layers share the same topology ($\mathcal{E}_P = \mathcal{E}_I = \mathcal{E}$) with unit weights $(\sigma_e)_P = (\sigma_e)_I = 1$. To incorporate heterogeneity, the parameters $\mu_i$ are sampled from a uniform distribution $[3.2, 4.8]$. 
The pinner, acting as a opinion leader with $\mu_{N+1} = 4$, attempts to steer the collective towards a target state by pinning 10 randomly selected nodes via five directed triadic interactions.
We denote with $\delta_i(t)=x_i(t)-x_{N+1}(t)$ the opinion mismatch of the $i$th node with respect to the pinner, and define the overall opinion mismatch as $\delta(t)=\|[\delta_1,\ldots,\delta_N]\|$.

\begin{figure*} 
    \centering
    \begin{subfigure}[b]{0.32\textwidth}
        \centering
        \includegraphics[width=1\textwidth]{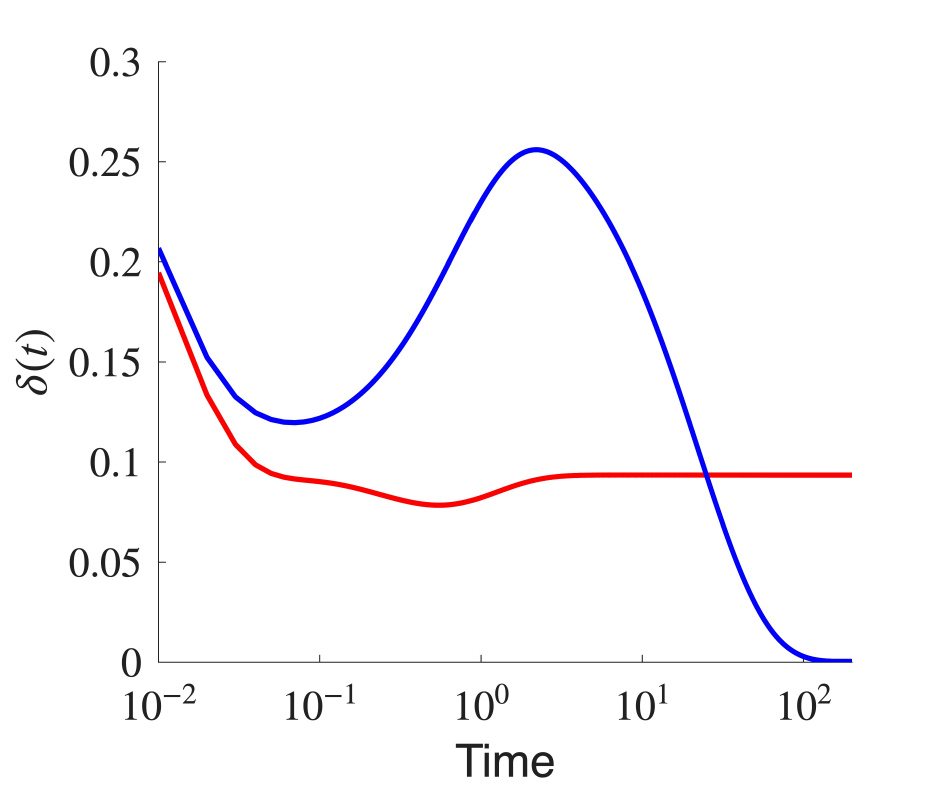}
        \caption{}
        \label{Fig.7.a}
    \end{subfigure}
    \hfill
    \begin{subfigure}[b]{0.32\textwidth}
        \centering
        \includegraphics[width=1\textwidth]{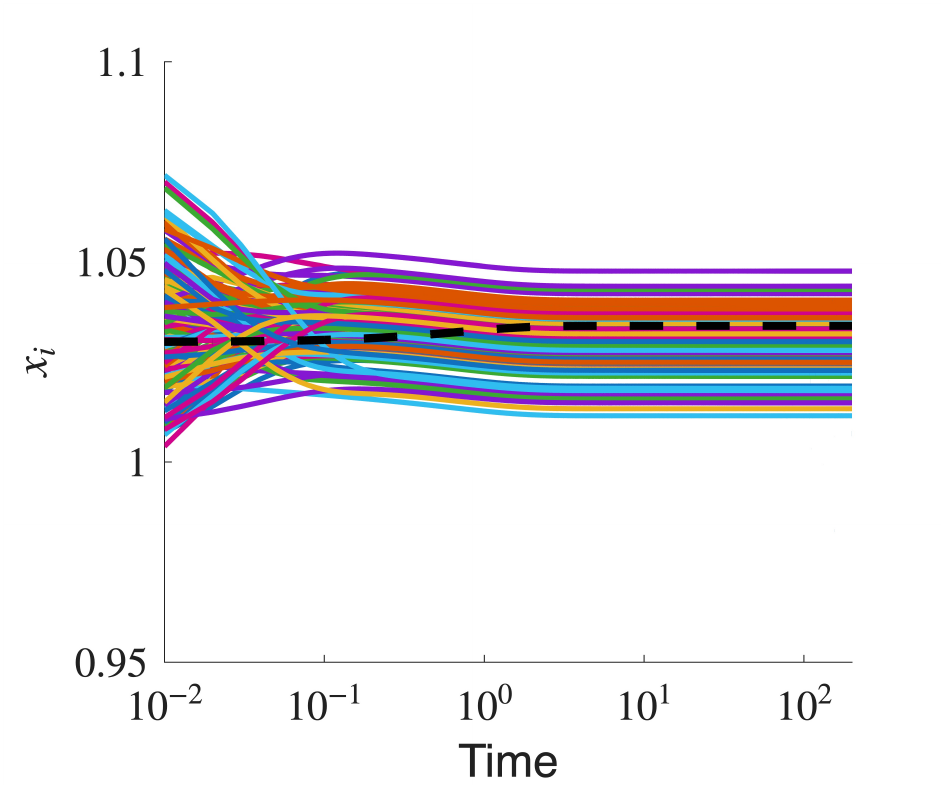}
        \caption{}
        \label{Fig.7.b}
    \end{subfigure}
    \begin{subfigure}[b]{0.32\textwidth}
        \centering
        \includegraphics[width=1\textwidth]{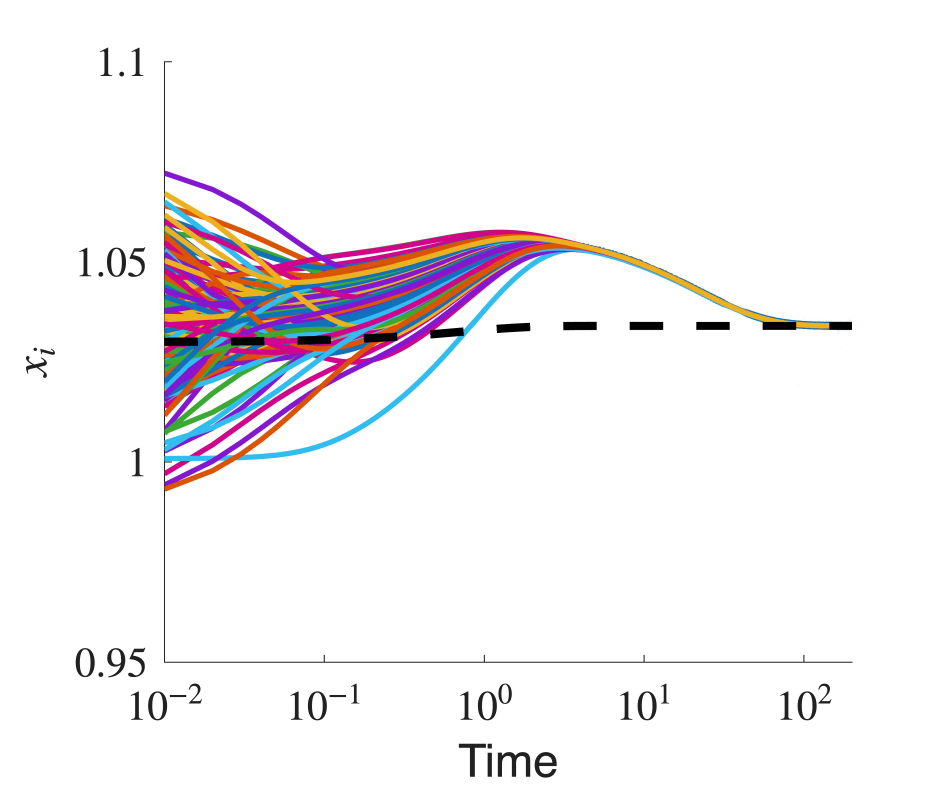}
        \caption{}
        \label{Fig.7.c}
    \end{subfigure}
    \caption{Opinion dynamics for a group of $N=150$ individuals according to the individual vector field \eqref{eq:opinion} coupled with the directed hypergraph of both the proportional and integral layer selected as the real-world social network from \cite{rizzello_pinning_2024}, and with a coupling protocol $g(x)=x$. Panel (a) reports the time evolution of the overall opinion mismatch $\delta(t)$ in the absence (red line) or in the presence (blue line) of the integral coupling layer. Panels (b) and (c) report the time evolution of the agent opinion in the absence and in the presence of the integral action, respectively, with a dashed black line identifying the opinion of the leader.}
    \label{Fig.7}
\end{figure*}

Fig.~\ref{Fig.7} illustrates the opinion dynamics under the influence of a pinning leader in the presence of dynamic diffusive coupling. 
In Fig.~\ref{Fig.7.a}, the temporal evolution of the overall opinion mismatch, $\delta(t)$, is compared for the cases of purely diffusive and dynamic (PI) diffusive coupling. It is evident that the inclusion of integral action effectively eliminates the persistent steady-state error inherent to purely diffusive coupling, and also reported in \cite{rizzello_pinning_2024}. The corresponding opinion trajectories $x_i(t)$ are presented in Figs.~\ref{Fig.7.b} and \ref{Fig.7.c}. While diffusive coupling results in a bounded but non-zero synchronization error, the dynamic coupling scheme facilitates asymptotic convergence of all agent opinions to the specific opinion of the pinner, thereby achieving complete leader-follower consensus.

Notice that, for the opinion dynamics model \eqref{eq:opinion}, 
\(J_\mu f(x,\mu)=\tanh(x)\) is not constant. Hence, the sufficient
condition used above to guarantee invariance along a generic
time-varying synchronous trajectory is not satisfied. This does not
contradict the numerical results in Fig. \ref{Fig.7}, since the pinning problem
considered here concerns convergence to a constant leader opinion.
At the limiting consensus equilibrium \(x_1=\ldots=x_{N+1}=s^\ast\), the
mismatch contribution \(\tanh(s^\ast)\delta\mu\) is constant and can
therefore be compensated by the integral action.

\section{Conclusion}

This work addressed the problem of achieving complete synchronization in directed hypergraphs of nonidentical nonlinear agents, where heterogeneity arises from parameter mismatches. In contrast with standard diffusive coupling, which under parametric mismatches generally renders a non-invariant synchronization manifold and therefore prevents exact synchronization, we showed that dynamic diffusive coupling can restore complete synchronization under precise structural conditions. In particular, the analysis revealed that complete synchronization is possible only when the parameter mismatches enter the transverse dynamics as constant additive disturbances, so that they can be compensated by the integral component of the coupling protocol. This highlights both the potential and the intrinsic limitations of dynamic coupling strategies in higher-order nonlinear networks.

By exploiting the correspondence between directed hypergraphs and signed graph representations, we derived necessary and sufficient conditions for the local invariance of the synchronization manifold. We then developed a Master Stability Function framework tailored to directed hypergraphs with proportional-integral hyperdiffusive coupling. This framework provides tractable criteria for assessing the local transverse stability of the synchronous solution in terms of the spectra of the Laplacian matrices associated with the proportional and integral layers. The results also clarify the role of the integral layer: beyond compensating constant mismatches, it can enlarge the synchronizability region and, in some cases, overcome destabilizing effects associated with negative real-part eigenvalues in the hypergraph Laplacian.

The theoretical predictions were validated through numerical simulations on directed hypergraphs of Lorenz oscillators subject to heterogeneous constant disturbances, where the MSF accurately identified the transition between synchronized and nonsynchronized regimes. Further simulations on a pinning-control problem in nonlinear opinion dynamics demonstrated that the proposed dynamic coupling scheme eliminates the persistent steady-state errors observed under purely diffusive coupling, enabling complete leader-follower consensus in the presence of nodal heterogeneity.

Overall, the results provide a systematic framework for analyzing and designing synchronization strategies in higher-order networks with dynamic diffusive interactions. Future work may extend the approach to broader classes of parameter mismatches, time-varying or adaptive hypergraph topologies, delays, stochastic perturbations, and more general forms of nonlinear coupling, as well as investigate less restrictive spectral conditions under which the MSF-based analysis remains applicable.


\section*{Acknowledgments}

A. Thomas Vadakkan  wishes to thank  Davide Salzano and Roberto Rizzello, University of Naples Federico II, for the useful discussions. During the preparation of this manuscript, the authors used AI-based tools to improve readability and language.

This work was supported by the European Union under Horizon Europe, grant number 101120085, Project “Higher-Order Networks and Dynamics (BEYOND THE EDGE)” - https://www.beyondtheedge.network/.








\printcredits

\bibliographystyle{elsarticle-num}

\bibliography{Reference}



\end{document}